\documentclass{article}

\usepackage{arxiv}

\usepackage[utf8]{inputenc} % allow utf-8 input
\usepackage[T1]{fontenc}    % use 8-bit T1 fonts
\usepackage{url}            % simple URL typesetting
\usepackage{booktabs}       % professional-quality tables
\usepackage{amsfonts}       % blackboard math symbols
\usepackage{nicefrac}       % compact symbols for 1/2, etc.
\usepackage{microtype}      % microtypography
\usepackage{lipsum}

\usepackage{relsize,balance,lipsum,bbm,enumerate,times,comment,color,graphicx,setspace,mathdots,mathrsfs,amssymb,latexsym,amsfonts,amsmath,cite,stmaryrd,caption,pgf,tikz,accents,mathtools,tabu,enumitem,hhline,array,epstopdf,nicefrac,amsthm,microtype,algorithmic,array,float,bm,url,subcaption}

\usepackage{graphicx}
\graphicspath{ {./images/} }

\newcommand{\V}{\mathcal{V}}
\newcommand{\B}{\mathcal{B}}

\usepackage{amsthm}

\DeclareMathOperator{\sgn}{sgn}

\title{Partisan Confidence Model for Group Polarization}

\author{
    Armineh~Rahmanian\\
    Dept. of Electrical \& Computer Engineering\\
    Tarbiat Modares University\\
    Tehran, Iran\\
    \texttt{armineh.rahmanian@modares.ac.ir}
\And
    Sadegh~Bolouki\\
    Dept. of Electrical \& Computer Engineering\\
    Tarbiat Modares University\\
    Tehran, Iran\\
    \texttt{bolouki@modares.ac.ir}
\And
    S.~Rasoul~Etesami\\
    Dept. of Industrial \& Enterprise\\
    Systems Engineering\\
    University of Illinois at Urbana-Champaign\\
    Urbana, IL, USA\\
    \texttt{etesami1@illinois.edu}
\And
    Abolfazl~Mohebbi\\
    Dept. of Mechanical Engineering\\
    Polytechnique Montr\'eal\\
    Montreal, QC, Canada\\
    \texttt{abolfazl.mohebbi@polymtl.ca}
}

\theoremstyle{definition}
\newtheorem{lemma}{Lemma}
\newtheorem{theorem}{Theorem}
\newtheorem{proposition}{Proposition}

\newtheorem{definition}{Definition}
\newtheorem{remark}{Remark}

\begin{document}
\maketitle
\begin{abstract}
    Models of opinion dynamics play a major role in various disciplines, including economics, political science, psychology, and social science, as they provide a framework for analysis and intervention. In spite of the numerous mathematical models of social learning proposed in the literature, only a few models have focused on or allow for the possibility of popular extreme beliefs' formation in a population. This paper closes this gap by introducing the Partisan Confidence (PC) model %based on a fundamentally different view of social influence
    inspired by the foundations of the well-established socio-psychological theory of groupthink. The model hints at the existence of a tipping point, passing which the opinions of the individuals within a so-called "social bubble" are exaggerated towards an extreme position, no matter how the general population is united or divided. The results are also justified through numerical experiments, which provide new insights into the evolution of opinions and the groupthink phenomenon.
\end{abstract}


%%%%%%%%%%
%%%%%%%%%%
\section{Introduction}\label{Introduction}
    Opinion dynamics is an important area of research with a wide range of applications in political campaigning, marketing, transportation management, public opinion management \cite{dong2018survey}, and group recommender systems \cite{RecommenderSystems}. In particular, due to the rapid growth of online social networks and unprecedented ease of opinion exchange on these platforms, there has been growing interest in how individuals' opinions are formed and perceived within a population \cite{mastroeni}. An interesting phenomenon frequently observed on these platforms, yet largely rejected by the existing opinion dynamics models, is that extremist positions can emerge and become mainstream \cite{Green, ElusiveConsensus, Mask, BrexitGeographic}. %BrexitGeographic

    In fact, the vast majority of opinion dynamics models in the literature, e.g., \cite{noorazar2020classical} and references therein, have been focused on the notion of {\it conformity}, which broadly refers to the tendency of an individual to act so as to fit into a group by adopting or touting what he or she perceives as the popular opinion. It is thus expected, as happens in virtually every existing conformity-based model, that conformist individuals avoid extreme beliefs and agree on or hover around a middle-ground, compromise position as time grows. Various types of cognitive biases, most notably those of the confirmation bias type, have been incorporated into conformity-based models to make them more realistic and better capture the evolution of opinions \cite{hegselmann2002opinion, Adaptive, EmergenceOfExtremism, Cyber-SocialNetworks}. \textit{Confirmation bias} in a broad sense refers to the tendency of an individual to actively seek to confirm her preconceptions. This may be done by avoiding or limiting exposure to beliefs unlike hers or selective perception of the facts presented to her, among other means. The inclusion of the confirmation bias in conformity-based models has effectively broadened research on global agreement scenarios to more general scenarios in which one or multiple consensus clusters could be reached.% {\color{blue} In addition, the general opinion in each cluster cannot be more extreme than the initial boundaries of opinion inside the cluster. Therefore, the formation of popular extreme beliefs based on on just bounded-confidence notion is unjustifiable.} 

    Some attempts have been made to modify classical conformity-based models to ones where opinion polarization is not beyond the realms of possibility. A notable example is the Altafini model \cite{altafini2012dynamics, altafini2012consensus}, in which the notion of antagonism among individuals is incorporated, leading to an influx of models that consider antagonistic relationships in opinion networks. Among these, some models have also included {\it bounded confidence} and {\it biased assimilation}, each of which can be viewed as a type of confirmation bias, as an individual's characteristics in the dynamics \cite{InnerChoice, PolarizationSigned,BiasedAntagonism}. While antagonistic interactions/relationships can contribute to and justify the tendency toward more extreme beliefs in a divided population, such as that of the United States with respect to political ideology (see Fig.~\ref{US polarization}, extracted from \cite{PewResearchCenter2017}), extreme beliefs have also been observed to form in fully collaborative environments \cite{AbortionRights}. %The emergence of such phenomena calls for an expansion and development of models that can shed light on social behaviors in everyday life.
    Furthermore, the pre-existence of extreme tendencies is often required at the onset of evolution of opinions for these models to explain how such tendencies become mainstream.
    
    \begin{figure}[t]
        \centering 
        \includegraphics[width=.5\linewidth]
        {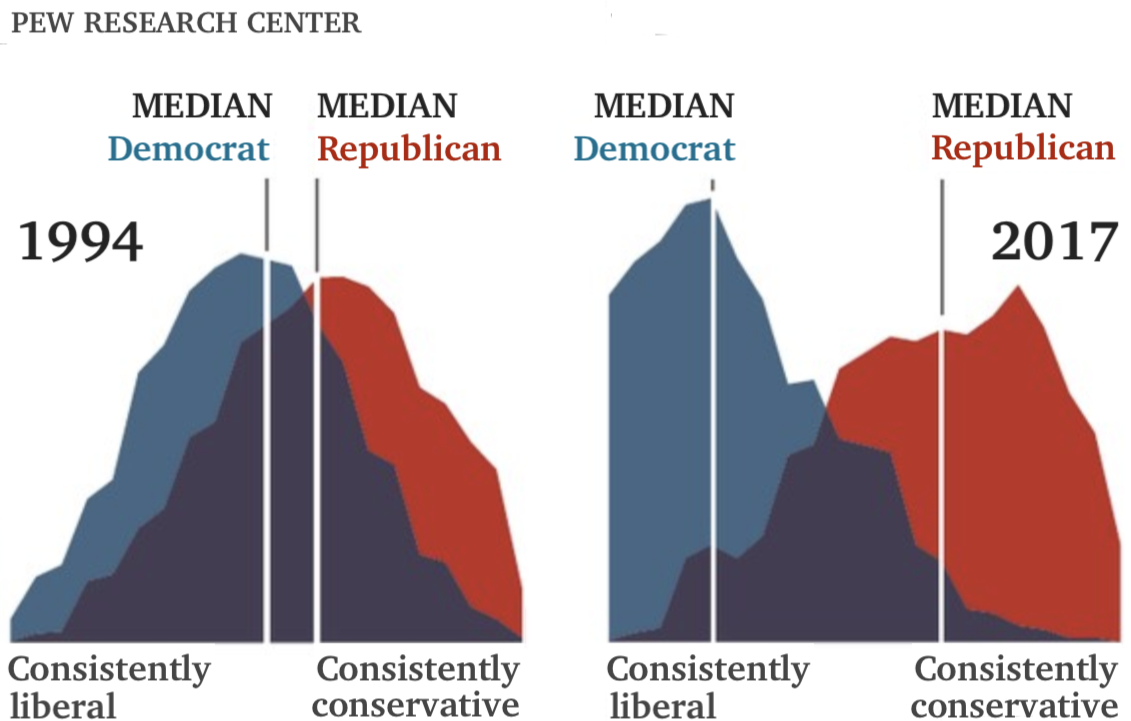}
        \caption{Growing ideological divide in the United States (Source: Pew Research Center \cite{PewResearchCenter2017}).}
        \label{US polarization}
    \end{figure}
    
    A handful of frameworks have been developed to address the emergence of popular extreme beliefs in fully collaborative environments where antagonistic relationships are absent. For instance, a model based on Persuasive Argument Theory (PAT) and homophily has been proposed in \cite{Differentiation, ArgumentBiasedAssimilation}, where it is suggested that those individuals who share similar viewpoints are more likely to engage in conversation with and have influence on each other. As another example, in \cite{BecomeExtremes}, a model has been introduced via the statistical physics modelling approach to explain the prevalence of extremism if stubborn extremists exist at the onset of the opinion evolution. Other remarkable examples are \cite{Dandekar2013, SocialBiases, AnalysisBiased}, where biased assimilation is incorporated into the Degroot's opinion averaging model \cite{degroot:74} to create the possibility of extreme beliefs emerging and becoming popular \cite{Dandekar2013, SocialBiases, AnalysisBiased}. %{\it biased assimilation}

    Our main objective in this work is to propose, conceptualize, and investigate a mathematical model inspired by the socio-psychological anlaysis of the groupthink phenomenon to study opinion formation. In particular, the model is capable of explaining the emergence of popular extreme beliefs in both united and divided populations. In the rest of section \ref{Introduction}, we provide more context to the problem of group polarization by further surveying the opinion dynamics literature and
    reviewing some of the basic concepts and notions from the groupthink theory and related models, before highlighting our contributions in this work and describing the organization of the paper.

%%%%%%%%%%

\subsection{Further Background to Opinion Dynamics}
\label{Intro.RelatedWork}
    
    Agent-based models of opinion dynamics study the evolution of individuals' opinions through interactions between them. In that regard, DeGroot in \cite{degroot:74} provides one of the basic models, where an individual's opinion is updated to a weighted average of the neighboring individuals' opinions, including the individual himself/herself. The interested reader is referred to \cite{bullo2018lectures} for a detailed analysis of the DeGroot model and its straightforward extensions using tools from algebraic graph theory. A modified version of the DeGroot model, the Friedkin-Johnsen model \cite{friedkin2011social}, takes into account susceptibility to persuasion and the degree of individuals' stubbornness. Also, an alternative polar model \cite{Polar,Nematzadeh} suggests that the susceptibility to persuasion of each individual is based on his/her current opinion. In addition to cooperative relationships, there may also be antagonistic relationships among individuals in a balanced structure, %representing trust/distrust or like/dislike,
    leading to a bipartite consensus in limit \cite{altafini2012dynamics,altafini2012consensus}, which could explain the shaping of a divided population. Sufficient conditions for bipartite consensus under time-varying collaborative and antagonistic relationships are studied in \cite{HostileCamps}.  
    
    %The interested reader is referred to \cite{bullo2018lectures} for a detailed analysis of the DeGroot model and its straightforward extensions using tools from algebraic graph theory.
    
    %In my opinion this sentence was not that relevant. >>>
    %A modified bounded confidence model that can preserve the sign of initial positions is studied in \cite{SignedBounded}. 
    The introduction of the Hegselmann-Kraus model  \cite{hegselmann2002opinion} has led to a class of nonlinear models wherein only the individuals who hold very similar opinions interact with and influence each other, in a manner that resembles confirmation bias, e.g., \cite{SignedBounded, DiscreteSignedBounded} and references therein. Such interactions may result in the emergence of multiple consensus clusters. A case in which antagonistic, neutral, and collaborative relationships are defined with respect to confidence levels is studied in \cite{BalanceSeeking}, where it is shown that under certain conditions for the amount of confidence levels for mentioned types of relationships, global consensus, bipartite consensus or clustering of opinions can be achieved. 
    % Finally, \cite{Dandekar2013, SocialBiases, AnalysisBiased} consider another nonlinear model with biased assimilation, in which an agent receives and takes into account new opinions by a bias toward them, i.e., reinforcing those new opinions which confirm his/her internal beliefs and diminishing non-confirming ones. It was shown that under certain conditions (e.g., a society consisting of agents with strong biases), divides among opinions may grow.

    All the aforementioned models consider the evolution of opinions without the existence of agreement pressure on individuals. However, there is evidence that the agreement pressure in a group often influences individuals' opinions. For instance, the study of a case in which individuals are exposed to increasing but bounded agreement pressure states that individuals' opinions will converge to a fixed distribution \cite{AgreementPressure}. A modified version of the Hegselmann-Krause model incorporating a pressure parameter reaches consensus more easily and faster than the original Hegselmann-Krause model \cite{cheng2019opinion}. A model accounting for inconsistency between private and expressed opinions due to conformity pressure is studied in \cite{ye2019influence}. This model can describe Asch's line-matching paradigm thoroughly with the help of resilience to pressure and susceptibility to persuasion parameters. We refer the reader to \cite{anderson2019recent} for other recent developments in the field of opinion dynamics.
    
%%%%%%%%%%
\subsection{An Overview of Groupthink}
    
    The theory of groupthink was first introduced by Irving L. Janis in 1971 \cite{janis1971groupthink}. By definition, groupthink is a concurrence-seeking tendency. When this tendency becomes dominant in a cohesive in-group, members will irrationally disregard and decline unpopular realistic views \cite{janis1971groupthink}. In other words, this tendency prevents people from treating controversial views realistically \cite{ElliotAronson2016Social}. Groupthink has a close relation to the Solomon Asch line-matching paradigm \cite{robbins2018essentials}. Members holding different opinions from the majority are under pressure to conform and suppress their true opinions. Groupthink does not always happen, as it requires certain antecedent conditions to be met, such as high cohesiveness within the group, isolation of the members from contrary views, and a lack of impartial leadership. If these antecedents are met, certain consequences will be observed, such as the illusion of invulnerability/unanimity and self-censorship, often leading to defective decision-making \cite{janis2008groupthink}. This process is summarized in Fig.~\ref{groupthink conditions}. Remarkable consequential examples of the groupthink phenomenon and victims of groupthink %, which led to fiascoes
    include the Bay of Pigs; the Pearl Harbor attack; the North Korean escalation; the Vietnam escalation \cite{janis1972victims}; the grounding of Swissair, the flying bank \cite{hermann2010grounding}, and more recently, to some degree, the failures at each stage of Brexit % are attributable to some degree of groupthink
    \cite{lees2020brexit}.

    \begin{figure*}[t!]
        \centering
        \includegraphics[width=1\linewidth]
        {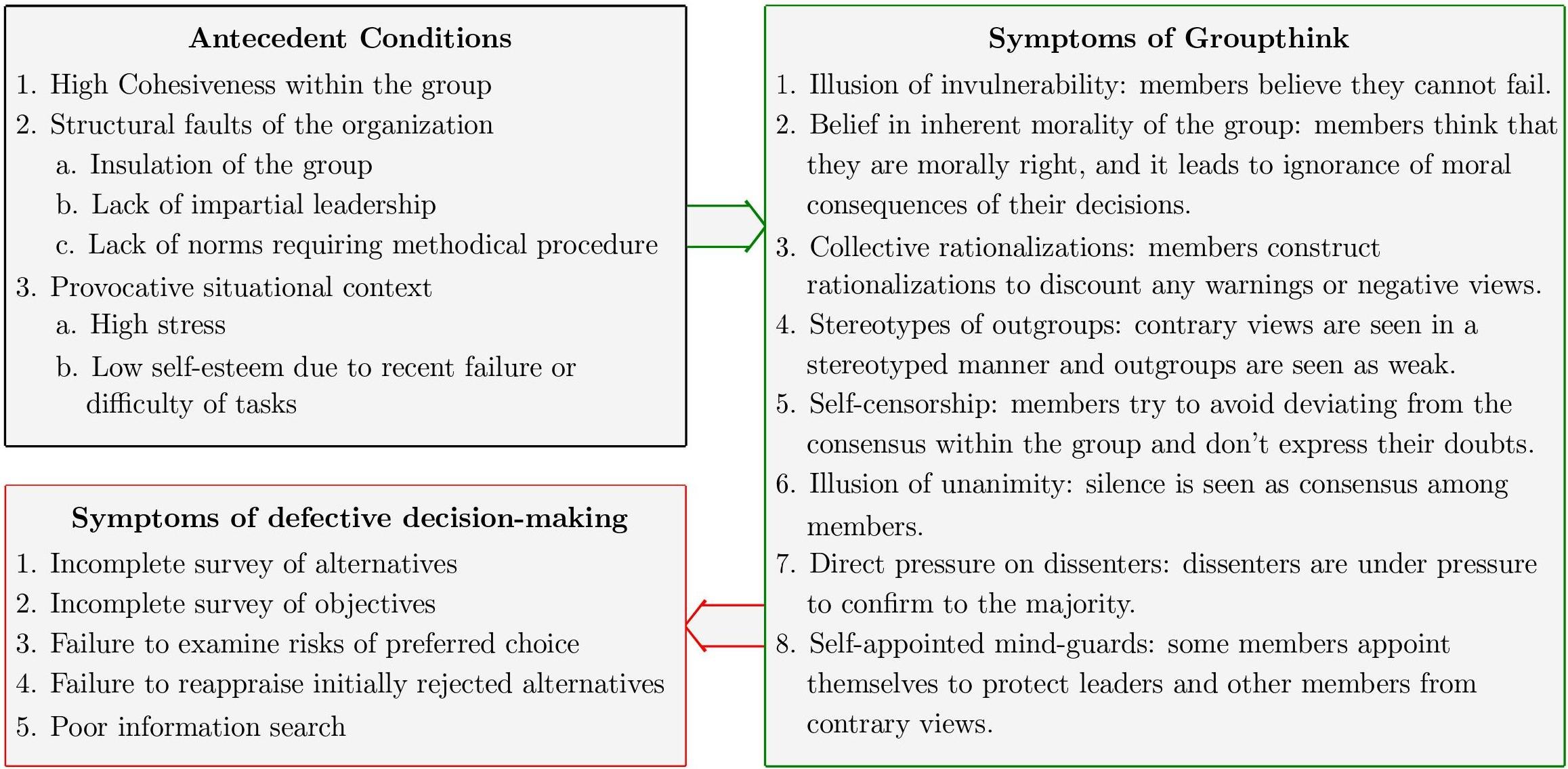}
        \caption{Groupthink phenomenon in summary, based on \cite{janis2008groupthink}}
        \label{groupthink conditions}
    \end{figure*}
    
    Following the introduction of the groupthink theory by Irving L. Janis, several modifications have been proposed \cite{rose2011diverse}. According to the so-called {\it ubiquity model} of groupthink \cite{baron2005so}, the antecedent conditions in Fig.~\ref{groupthink conditions} are not necessary for the groupthink to occur. In fact, there is evidence that the symptoms in Fig.~\ref{groupthink conditions} have arisen without the presence of such antecedents. The ubiquity model introduces three key conditions necessary for this phenomenon to happen in everyday groups. They include social identification, salient norms, and low self-efficacy. Social identification is associated with social acceptance and social rejection, which means that a deviating member will be rejected socially as a punishment. The second antecedent, salient norms, suggests that a group norm will be created by interaction and discussion within the group. The last antecedent, low self-efficacy, is related to situational lack of self-confidence in one's problem-solving ability due to recent failure, fear, time pressure, or other factors. All three key conditions act as a motivation for suppressing dissenting views despite the realistic spirit of their content.

    {\it Group polarization}, as a direct, immediate consequence of the groupthink phenomenon, is a group's tendency to make decisions that are more extreme than the initial positions of its members \cite{ElliotAronson2016Social}. For examples, the group's decision moves toward greater risk (or greater caution) if the initial tendency is to be somewhat risky (or cautious). The term group polarization, with which the current work is mainly concerned, should not be confused with opinion polarization, even though the two phenomena are correlated and may coexist within a population.
    %Groupshift, however, is defined as the change between the agent's decision and group's decision \cite{robbins2018essentials}. This shift is usually toward the extreme version of the group's original position, which is called risky (or cautious) shift. 

%%%%%%%%%%

\subsection{Contributions}
\label{contributions}

    The main contribution of this paper is the introduction and investigation of the Partisan Confidence (PC) model for social learning. It is capable of justifying the emergence of popular, extreme beliefs in a social network. The PC model is inspired by socio-psychological analysis of the groupthink phenomenon, the presence of which is evident in both united and divided populations. The PC model is based on a unique understanding of opinions and social influence in that opinions of like-minded individuals tend to resonate when they interact, compared to most of the existing models wherein a compromise instead of resonance is in play. Furthermore, under the PC model, the pre-existence of extremists for the emergence of extreme beliefs is not necessary. A detailed description of the paper's contributions is stated below.
    \begin{itemize}[leftmargin=.125in]
        \item {\it Decomposition of opinions based on Partisan Confidence:} An individual's opinion at any given instant is represented by a scalar, which is consistent with the vast majority of the existing opinion dynamics models. However, we shall rely on a unique characterization of opinions sketched as follows. Assuming that an individual's opinion lies in the interval $[-1,1]$, its sign is viewed as her general belief (ideology, political party, etc.), while its magnitude is interpreted as her confidence about that general belief, justifying the term {\it Partisan Confidence} It should be noted that the partisan confidence decomposition of opinions has resemblances to the viewpoints leading to the models presented in \cite{Dandekar2013,EmergenceOfExtremism}. In \cite{EmergenceOfExtremism}, it is assumed that in some situations the exact opinion of an individual is not known while its discrete choice (yes, no, neutral) is expressed publicly. In \cite{Dandekar2013}, the distance between an individual's opinion and each endpoint of the opinion spectrum represents her support for that marginal belief. %{\color{blue} It should be noted that our representation of opinion is different to what has been incorporated in both Danker's {\etal}  \cite{Dandekar2013} and Kangqi's {\etal} model \cite{EmergenceOfExtremism}. In the latter, the distance to the boundaries of opinions represents the support for each end of the opinion spectrum. While in the former, which is based on bounded confidence and negative influence, it is supposed that in some situations the exact opinion of an agent is not known while its discrete choice is expressed publicly. Thus, this choice, which can be yes, neutral or no, plays a direct role in forming an individual's inner opinion. In contrast, in ours, which will be discussed in detail in the following sections, the sign of an agent's opinion is a reflection of her political orientation or partisanship.} % choice -A or +A. 
        %%%
        \item {\it Social influence characterization:} In our approach, when individuals interact, influence is modeled in such a way that (i) the opinions of like-minded individuals, i.e., those who share the same general belief at the time, resonate with each other, meaning that the individual make one another more confident in that shared general belief; (ii) it accounts for the confirmation bias, which means that the influence on each other of individuals with opposite general beliefs is discounted; and (iii) an individual with higher confidence in her general belief is deemed more influential. The first item listed above defies the mainstream line of opinion dynamics research where like-minded individuals tend to compromise (not resonate) when they interact.% In other words, an individual with greater confidence in a general belief will not gain any confidence, and even loses some, when she interacts with an individual with less confidence in the same general belief.
        %%%
        \item {\it Introducing the PC and PC-lite models:} Based on the characterization of social influence described above, we introduce and investigate two models for social learning, namely the PC and PC-lite models. While the PC model is more compelling, since it also accounts for the confirmation bias (item (ii) above), the PC-lite model is of great importance as it demonstrates that group polarization may occur even without the contribution of confirmation bias, which highlights a fundamental difference between this work and \cite{Differentiation,Dandekar2013}. It is also important to note that the PC and PC-lite models, unlike those introduced in \cite{Differentiation,Dandekar2013}, are both time-varying, that is a must-have property for models capturing a practical social dynamics.
        %%% 
        \item {\it Deriving conditions for group polarization:} Through rigorous analysis of the PC and PC-lite models, we argue for the possible existence of communities within a population, henceforth called {\it social bubbles}, that are bound for group polarization if the antecedent conditions of groupthink are satisfied. A definition of {\it social bubble} is given; intuitively, it is a group isolated from the outside population beyond a certain level. We then prove that if a certain degree of homogeneity is present within a social bubble at some point in time, that is, if the individuals in the bubble share the same general belief and are confident in that belief beyond a relatively low degree, then their confidence will increase over time and reach an extremely high value, at which point the group polarization will have seemed to materialize in this bubble. The extent of the bubble's polarization is shown to be directly related to its isolation level, and for the PC model, also to the intensity of the confirmation bias of those within the bubble. %{\color{blue} More importantly, the reliance of a bubble's polarization level to mentioned factors gives both the polarization level and the tipping point a degree of freedom that sounds more realistic rather than ultimate polarization - complete degree of support or disapproval - and also a fixed tipping point which is enforced by other models \cite{Differentiation, Dandekar2013}.}
        
        % \item {\color{blue} {\it Extending the conditions for group polarization to the general case:} The shape of  confirmation bias associate with a given society was not restricted to a special definite shape. In fact, by introducing an upper bound that stands for the shape of confirmation bias at play in a given society, we were able to extend the conditions and results of group polarization to a more general case as well. 
        % } 
        
        %} and \emph{bubble number} to capture a cohesive group and the insulation degree of a group, respectively. We then provide a sufficient condition on opinions in a social bubble for occurring groupshift. The proposed model, together with the sufficient condition of occurring groupshift, form the groupthink model based on the classic theory. However, in order to capture extra features of the groupthink, we enrich our model by relaxing the degree of cohesiveness and isolation of a group using a discounting factor. We also provide sufficient conditions on opinions for occurring groupshift in the modified model with a discounting factor. The modified model, together with the sufficient condition of groupshift occurrence, form ubiquity model of groupthink.  Finally, we support our theoretical results using numerical experiments. 
    \end{itemize}  

%%%%%%%%%%
\subsection{Paper Organization}

    %The remainder of this paper is organized as follows.
    In Subsection~\ref{Preliminaries And Notation}, we provide some preliminaries and notations for later use. In Section~\ref{First Model and its properties}, we introduce and investigate a basic model, the so-called {\it PC-lite} model, generalization of which leads to the PC model in Section~\ref{Second Model and its properties}. We discuss the inherent properties of the proposed models in their corresponding sections. In Section~\ref{Numerical Examples}, we provide simulations to demonstrate properties of the proposed models. We conclude the paper by identifying some future directions of research in Section~\ref{Conclusion}. For ease of presentation, we relegate all the proofs and auxiliary lemmas to the Section \ref{Proofs}.

%%%%%%%%%%
\subsection{Notation}
\label{Preliminaries And Notation}

    We let $\V = \{1,\ldots,n\}$ be the set of all individuals or, as is often called from now on, {\it agents.} The opinion of agent $i$ at discrete time $t$, $t \geq 0$, is denoted by $x_i(t) \in [-1,1]$, and its sign is denoted by $\sgn(x_i(t))$. A weighted, time-varying digraph $\mathcal{G}(t)=(\V,\mathcal{E}(t),W(t))$ is assumed to represent the topology of the network of agents over time, where $\V$ is the set of nodes, $\mathcal{E}(t) \subseteq \{(i,j)|i,j \in \V, i\not=j\}$ is the set of edges at time $t$ indicating the interactions among the agents, and an element $w_{ij}(t)$ of the weight matrix $W(t)$ indicates the weight of influence of agent $j$ on agent $i$ at time $t$. It is assumed that $w_{ij}(t)>0$ if $(i,j) \in \mathcal{E}(t)$, and $w_{ij}(t)=0$ otherwise. The non-negative matrix $W(t)$ is called {\it row-stochastic} if the elements of each of its rows sum up to 1, that is, if $\sum_{j\in \V} w_{ij}(t)=1$ for all $i \in \V$. It is called {\it row-substochastic} if $\sum_{j\in \V} w_{ij}(t)\leq 1$ for all $i \in \V$, and there exists some $k\in \V$ such that $\sum_{j\in \V} w_{kj}(t) < 1$. Throughout the paper, and in the proofs in particular, the argument $t$ of time-varying functions is dropped for the benefit of notational convenience. For instance, $x_i$, $z$, and $w_{ij}$ often replace $x_i(t)$, $z(t)$, and $w_{ij}(t)$, respectively. Furthermore, an agent's update value, $x_i(t+1)-x_i(t)$, will be denoted by $\Delta x_i$, that itself is short for $\Delta x_i(t)$. 

    %Finally, we denote the space of $m\geq 2$ dimensional vectors with coordinates in $\{-1,1\}$ by $\mathbb{E}^m$.

%%%%%%%%%%
%%%%%%%%%%
\section{Partisan Confidence-lite Model and Its Properties}
\label{First Model and its properties}

    In this section, we introduce, justify, and investigate the Partisan Confidence-light (PC-lite) model for the evolution of opinions in a social network. Let $x_i(t) \in [-1,1]$ denote the opinion of agent $i \in \V$ at discrete time $t \geq 0$, where $\V = \{1,\ldots,n\}$ is the set of all agents. In the PC-lite model, the opinion of every agent $i$ evolves according to the following discrete-time dynamics:
    \begin{flalign}
        \text{\bf(PC-lite dynamics)}\hspace{1.1in}        \Delta x_i = \sum_{j \neq i} \Big[w_{ij} |x_j| (\sgn(x_j) - x_i)\Big].&&
    \label{model}
    \end{flalign}
    To be clear, as described in Subsection~\ref{Preliminaries And Notation}, the dynamics \eqref{model} should be read as
    \begin{equation}
        x_i(t+1) - x_i(t) = \sum_{j \neq i} \Big[ w_{ij}(t) |x_j(t)| (\sgn(x_j(t))-x_i(t)) \Big].
    \end{equation}
    According to the PC-lite dynamics \eqref{model}, a self-weight $w_{ii}$ is not present and does not contribute to the opinion change of agent $i$ at time $t$. Thus, we can assume $w_{ii}(t)=0, \forall i\in\V, t\ge 0$. This assumption eliminates the self-loop from the underlying graph $\mathcal{G}(t)$, and hence results in a row-substochastic adjacency matrix $W(t)$.
%%%%%%%%%%
\subsection{Justification of the PC-lite Model}
\label{Justification of the PC-lite Model}    
    
    One notices that the PC-lite dynamics \eqref{model} follows, in principle, the same rule of social influence as the time-varying version of the DeGroot model \cite{degroot:74},
    \begin{flalign}
        \text{\bf(DeGroot dynamics)}\hspace{1.35in}        \Delta x_i = \sum_{j \neq i} \Big[w_{ij} (x_j - x_i)\Big].&&
    \label{degroot}
    \end{flalign}
    However, \eqref{model} is set up on a fundamentally distinctive interpretation of opinion perception, that is, agent $i$ perceives the opinion $x_j$ of agent $j$ as an approval of $\sgn(x_j)$ with confidence level $|x_j|$. Subsequently, the weight of influence of agent $j$ on agent $i$ is discounted by the factor $|x_j|$, while the opinion $x_j$ of agent $j$ is replaced by $\sgn{(x_j)}$. %In other words, agent $i$ views the opinion of agent $j$ simply as $\pm 1$ or zero.
    Therefore, $x_j$ is decomposed into two parts: a direction part $\sgn{(x_j)}$, which can be viewed as party affiliation in political terms, and an intensity or confidence part $|x_j|$. It is this decomposition of agents' opinions that justifies the appellation {\it Partisan Confidence}. Indeed, $\sgn{(x_j)}$ is often concerned with a much more specific issue than one's political party. For instance, on the issue of abortion rights, it addresses whether a person generally supports or is against abortion rights.

    The consideration that the influence weights $w_{ij}{(t)}$ are time-varying adds to the practicality of the PC-lite model since (i) no agent has to interact with the same set of agents at all times, i.e., there are asynchronous interactions, and (ii) the dynamics \eqref{model} with fixed weights $w_{ij}$ cannot accurately model human thinking, i.e., there is model uncertainty.

    We may use $w_{ij}(t)|x_j(t)|$ to refer to the overall influence of agent $j$ on agent $i$ at time $t$. The amount of this overall influence is determined by the intensity of an opinion $|x_j(t)|$. On the other hand, the direction $\sgn{(x_j(t))}$ determines whether the overall influence is in favor of or against an opinion.
    In fact, one can think of $|x_j(t)|$ as the intensity of the emotion being transferred to an another agent that either advocates or disapproves of some position or idea, and this intensity governs the overall influence. Thus, the more extreme the emotion, the greater the overall influence would be, and vice versa. When the opinion of agent $j$ at time $t$ is zero, we assume that she is completely neutral; thus, her opinion will not drag the opinion of agent $i$ at time $t$ in either directions on the opinion spectrum. In other words, a neutral opinion dose not contribute to the opinion change of an agent. %$i$.

%%%%%%%%%%
\subsection{Properties of the PC-lite Model}
\label{Properties of the First Model}

    We now investigate the PC-lite dynamics \eqref{model} in detail and discuss why it can explain the groupthink behavior. A key antecedent condition for groupthink is isolation of the group members from the outside population. We start off with a simple but important result that highlights the unique capability of the PC-lite model (and also the PC model discussed in the next section) in explaining group polarization and the emergence of popular, extreme beliefs in a network of agents.

    \begin{definition}[\bf{Connectedness}]
        Given the model \eqref{model}, a subset $\B\subseteq \V, |\B| > 1$, of agents is said to be {\it connected} if
        \begin{equation}
            \sum_{j \in \B}\sum_{t=0}^\infty w_{ij}(t) =\infty,~\forall i \in \B.
        \end{equation}
    \end{definition}
    
    A connected subset $\B$ of individuals is not necessarily blended, a property defined as strong connectivity of the graph with nodes $\B$, $|\B|>1$, and edges $(i,j)$ which exist if and only if
    \begin{equation}
        \sum_{t=0}^\infty w_{ij}(t) =\infty.
    \end{equation}
    More precisely, a subset is connected if and only if it is blended or can be divided into multiple blended subsets. This means that any result stated later on for connected subsets also applies to blended subsets, which may be of greater interest in the given context of social networks.
    
    \begin{proposition}
    \label{prop:polarization_result}
        Given the model \eqref{model}, let a connected subset $\B\subseteq \V$ of agents be isolated from outside, i.e., for any $i \in \B$ and $t \geq 0$, assume that
        \begin{equation}
            \sum_{j \in \V \backslash \B} w_{ij}(t)=0.
        \label{prop 1 condition}
        \end{equation}
        If for some $t_0$, $x_i(t_0) > 0$, $\forall i \in \B$, then we have $\lim_{t \rightarrow \infty}{x_i(t)} =1$, $\forall i \in \B$.
    \end{proposition}    
    
    \begin{proof}
        Proposition~\ref{prop:polarization_result} will turn out to be a special case of Proposition~\ref{polarization_result}, stated later on in this section. See Subsection~\ref{rest of proofs} for more details.
    \end{proof}
    
    Proposition~\ref{prop:polarization_result} addresses an exaggerated but important situation in which a set of connected agents is completely isolated from the rest of the population. It states that if the agents in this set are initially in agreement, no matter how weak this agreement is, they reach an extremely strong agreement as time passes. In other words, in the absence of opposing views, given consistent interactions among the agents, group polarization is inevitable.
    
    In the following, we argue that the complete isolation condition for group polarization can be relaxed to a realistic one via the notion of a {\it social bubble} that may be present in a population, defined below based on the notion of a {\it bubble number}.

    \begin{definition}[\bf{Bubble number}] Given the model \eqref{model} and a subset $\B\subseteq \V, |\B| > 1$, of agents, the bubble number of $\B$, denoted by $\gamma_\B$, is defined as the largest non-negative constant $\gamma$ that satisfies the following equation for any $i \in \B$:
        \begin{equation}
        \label{buble.prop1}
             \sum_{j \in \B} w_{ij}(t) \geq \gamma_\B \sum_{j \in \V \backslash \B} w_{ij}(t), ~ \forall t.
        \end{equation}
    The bubble number is well-defined for any $\B$ since \eqref{buble.prop1} is satisfied by an upper-bounded, closed interval in $\mathbb{R}$ containing 0.
    \end{definition}

    Equation \eqref{buble.prop1} states that the sum of the weights inside $\B$ is at least $\gamma_\B$ times larger than the sum of the weights to agents outside that bubble. Thus, $\gamma_\B$ indicates the isolation level of $\B$ from outside in the sense of opinion influence. The greater the bubble number $\gamma_\B$, the greater the isolation of the members in the subset. It is worth noting that the bubble number is closely related to the so-called \emph{cut ratio} of a weighted graph \cite{godsil2013algebraic}. More precisely, if we sum \eqref{buble.prop1} over all $i\in \B$, we obtain
    \begin{equation}
        \frac{1}{\gamma_{\B}}\ge \frac{\sum_{i\in \B, j \in \V \backslash \B} w_{ij}(t)}{\sum_{i,j \in \B} w_{ij}(t)},
    \end{equation}
    where the expression on the right side is the ratio of the sum of the edge weights crossing the cut $\B$ over the sum of the edge weights inside that cut. It is known that the minimum cut ratio over all the cuts can be bounded by the algebraic connectivity of the graph \cite{godsil2013algebraic}. Therefore, one can bound the bubble number in terms of the eigenvalues of the adjacency matrix $W(t)$.%{\color{green} (of $\mathcal{G}$), should be cleared after checking.} % However, we note that the bubble definition is stronger than the cut ratio in the sense that it requires the inequality \eqref{buble.prop1} to hold for any $i\in \B$ and not just for the sum of the agents in $\B$.

    \begin{definition}[\bf{Social bubble}]
    \label{soc bub}
        Given the model \eqref{model}, a subset $\B\subseteq \V, |\B| > 1$, of agents is loosely called a {\it social bubble}, or simply a {\it bubble}, if it has a large bubble number, meaning that it is, to a great extent, isolated from outside influence.
    \end{definition}

    % \begin{definition}[\bf{Group Polarization}]
    %     We say that a social bubble experiences {\it group polarization} if a weak agreement within the bubble leads to a strong one as interactions continue.
    % \end{definition}
    
    From the theory of groupthink, it is expected that agents in a connected bubble will intensify cohesiveness (if it exists) in a discussion, seeking stronger agreement within the bubble. According to the following theorem, this phenomenon is very well captured by the PC-lite model for social learning.
    
    \begin{proposition}
    \label{polarization_result}
        Given the PC-lite dynamics \eqref{model}, let a connected subset $\B\subseteq \V$ of agents have the bubble number
        \begin{equation}
        \label{low1}
            \gamma_\B>3+2\sqrt{2},
        \end{equation}
        and assume that $\alpha_1$ and $\alpha_2$, where $\alpha_1 < \alpha_2$, are the two positive solutions of the equation
        \begin{equation}
        \label{alphagamma}
            \frac{1+\alpha}{\alpha(1-\alpha)} = \gamma_\B.
        \end{equation}
        If, for some $t_0$, it happens that
        \begin{equation}
        \label{alpha1}
            x_i(t_0) > \alpha_1,~\forall i \in \B,
        \end{equation}
        then we have
        \begin{equation}
        \label{alpha2}
            \liminf_{t \rightarrow \infty}{x_i(t)} \geq \alpha_2,~\forall i \in \B.
        \end{equation}
    \end{proposition}
    
    \begin{proof}
        Proposition~\ref{polarization_result} will turn out to be a special case of Proposition~\ref{polarization_result_Second_model}, stated later on in the paper. See Subsection~\ref{rest of proofs} for more details.
    \end{proof}
    
    The threshold $3+2\sqrt{2}$ in Proposition~\ref{polarization_result} marks the smallest possible $\gamma_\B$ for which \eqref{alphagamma} has two positive real solutions for $\alpha$. It can also be viewed as the threshold that makes the loosely defined notion of a ``social bubble'' in Definition \ref{soc bub} precise. Therefore, Proposition~\ref{polarization_result} implies that if the agents in a bubble reach a certain degree of cohesiveness, that is, are at least $\alpha_1$-confident in advocating in favor of a common position, then their confidence tends to grow higher, beyond degree $\alpha_2$. As the bubble number $\gamma_\B$ increases, $\alpha_1$ and $\alpha_2$ will decrease and increase, respectively, as demonstrated in Figure~\ref{Alpha, model 1}. In limit, as $\gamma_\B$ goes to infinity, $\alpha_1$ reaches 0 while $\alpha_2$ reaches 1, making the case for Proposition~\ref{prop:polarization_result}.  For $\gamma_\B \simeq 7.83$, $\alpha_1$ and $\alpha_2$ are 1/2 apart. %Thus, the satisfaction of Proposition~\ref{polarization_result} implies an easier group polarization occurring in each bubble. That means that a smaller level of support or disapproval of a topic can lead to a higher-level one, which is more extreme.
    It should also be noted that the same can be said about a bubble in which the agents disapprove of a position. Hence, in summary, Proposition~\ref{polarization_result} shows group polarization occurring within a connected bubble if the opinions of the agents in the bubble have reached a certain degree of support/disapproval of any given position at some time $t_0$.
    
    We note that the term social bubble is broader than what is loosely known as {\it echo chamber,} in that a social bubble, unlike an echo chamber, may include individuals with diverse or even opposite beliefs. Proposition~\ref{polarization_result} demonstrates that once a social bubble turns into an ``echo chamber,'' i.e., once the condition \eqref{alpha1} is satisfied, one should expect exaggeration of those beliefs as time grows, that is \eqref{alpha2}.
    
     \begin{figure}[t]
         \centering
         \includegraphics[width=.5\linewidth]
         {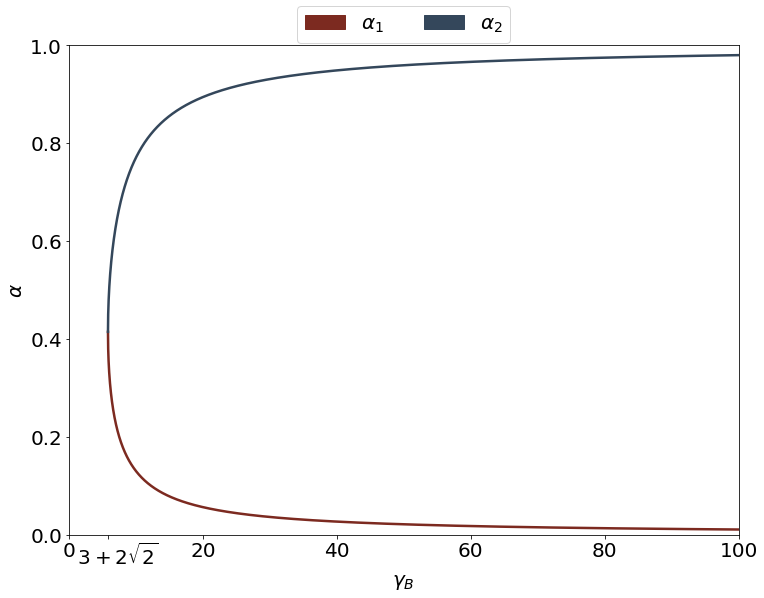}
         %corollary.PC-lite.2.png
         \caption{Changes of $\alpha_1$ and $\alpha_2$ with respect to $\gamma_B$}
         \label{Alpha, model 1}
     \end{figure}    
    
    \begin{remark}
    \label{corollary 1}  
        Suppose that $\V$ contains at least $m$ connected, pairwise disjoint social bubbles $\B_1, \ldots,\B_m$. Now, depending on whether, for each bubble $\B_k, k=1,\ldots,m$, we have $x_i(t_{0_k}) > \alpha_1,~\forall i \in \B_k$ or $x_i(t_{0_k}) < -\alpha_1,~\forall i \in \B_k$, where $t_{0_k} \geq 0$, we have
        \begin{equation}
            \liminf_{t \rightarrow \infty}{x_i(t)} \geq \alpha_2,~\forall i \in \B_k, \label{pos}
        \end{equation}
    or
        \begin{equation}
            \liminf_{t \rightarrow \infty}{x_i(t)} \leq -\alpha_2,~\forall i \in \B_k, \label{neg}
        \end{equation} 
        respectively. In particular, the asymptotic structure of the bubbles can be represented via one of the $2^{m}$ vectors $s\in \{-1,1\}^m$ such that $s_k=+1$ if \eqref{pos} holds, and $s_k=-1$ if \eqref{neg} holds. Thus, the network can exhibit at least $2^m$ substantially different limiting behaviors.
    \end{remark}
    
    % \begin{lemma}\label{increasing_model}
    %     The model \eqref{model} is non-decreasing, that is for any $y^1,y^2 \in [-1,1]^n$, we have
    %     \begin{equation}
    %         y^1 \leq y^2 \Rightarrow f(y^1) \leq f(y^2).
    %     \end{equation}
    % \end{lemma}
    
%%%%%%%%%%
%%%%%%%%%%
\section{Partisan Confidence Model and Its Properties}
\label{Second Model and its properties}

    Acting toward opinions with a bias has been well documented in confirmation bias theory \cite{breuning2020neurochemistry}. {\color{magenta}A}gents tend to respond with a bias toward information inconsistent with their own information, beliefs, and old experiences. Also, agents tend to willingly ignore some nonconforming information and opinions only to fit into their social groups \cite{allen2020mental}. In summary, this bias can be due to receiving information that inconsistent or in conflict with one's social norm or identity.% For instance, when one is not confident in his ability to solve or judge a problem, he instead sticks with the majority's opinion or the social norm.
    
    In this section, we introduce the Partisan Confidence (PC) model, which is a generalization of the PC-lite model \eqref{model} that accounts for the the agents' confirmation bias. As we discussed earlier, the PC-lite model \eqref{model} can describe the group polarization caused by the groupthink behavior described in Irving L. Janis's seminal work \cite{janis2008groupthink}. To fit that model into Robert S. Baron's more advanced model of groupthink \cite{baron2005so}, we assume that each agent $i$ discounts the influence of contrary views received from any other agent and propose the following opinion dynamics model:
    \begin{flalign}
    \label{ModifiedModelPart2}
        \text{\bf(PC dynamics)}\hspace{1.2in}               \Delta x_i= \sum_{j \neq i} \Big[d_i(x_i,x_j) w_{ij} |x_j| (\sgn(x_j) - x_i)\Big],&&
    \end{flalign}
    where $d_i:[-1,1]^2 \rightarrow [0,1]$ is a discounting function elaborated in the following subsection, before investigating the properties of the PC model. Inclusion of the discounting function in the PC model \eqref{ModifiedModelPart2}, that is, discounting of opposing views, in a sense amplifies the isolation degree of a cohesive bubble. Hence, in view of Proposition~\ref{polarization_result}, one expects that the bubble number threshold for group polarization should now be lower, as is made concrete later on.
    
\subsection{Discounting Function}
    
    As implied from its title, the discounting function is assumed to always return a number within $[0,1]$; a trivial assumption which will not be repeated but made throughout. Furthermore, in view of the confirmation bias, an agent $i$ is to discount the influence on her of an agent $j$ with a general belief opposite to hers. No discount is expected otherwise, meaning that
    \begin{equation}
        d_i(x_i,x_j) = 1 ~ \text{if}~ \sgn(x_i)=\sgn(x_j).
    \label{d 1}
    \end{equation}
    We also assume that the discount value for opposing general beliefs is at all times upper bounded as
    \begin{equation}
        d_i(x_i,x_j) \leq \hat{d}_i(|x_i|)~\text{if}~\sgn(x_i)\neq\sgn(x_j),
        \label{relaxed d}
    \end{equation}
    where $\hat{d}_i:[0,1]\rightarrow [0,1]$ is an arbitrary non-increasing function. It should be noted that the non-increasing assumption on $\hat{d}_i$ is reasonable, as it implies that confirmation bias increases with confidence. While our analysis shall remain valid for any discounting function satisfying \eqref{d 1} and \eqref{relaxed d} for a non-increasing $\hat{d}_i$, to shed some light on the PC model, we consider the following candidate for $\hat{d}_i$:
    \begin{equation}
        \hat{d}_i(|x_i|) = 1-(1-d)|x_i|^{\beta}
        \label{discounting function- alpah and beta}
    \end{equation}
    where $d$ and $\beta$ are constants satisfying $0 \leq d \leq 1$ and $\beta > 0$. This means that the condition \eqref{relaxed d} is now translates to    \begin{equation}
        d_i(x_i,x_j) \leq 1-(1-d)|x_i|^{\beta}~\text{if}~\sgn(x_i)\neq\sgn(x_j).
        \label{bound for general discounting function}
    \end{equation}
    In what follows, the interpretation of the parameters $d$ and $\beta$, along with the justification of the upper bound assumption, are given. The case for a general discounting function satisfying \eqref{d 1} and \eqref{relaxed d} will be discussed in the very end of the section.

Let us start with the reason why \eqref{bound for general discounting function} only imposes an upper bound on the discounting function instead of assuming an exact formulation. We believe that any exact formulation is too restrictive and unrealistic in a social network setting. An upper bound, with two degrees of freedom in $d$ and $\beta$, allows for a great deal of uncertainty and agents variability in the model, which means the results derived based upon the PC dynamics remain credible in a practical setting. It also addresses the case where the discounting function also varies over time, that is if $d_i$ is a function of $t$ besides $x_i$ and $x_j$.

With the exact discounting function approach ruled out, one wonders why upper-bounding is selected for approximating the discounting function among various possible non-exact formulations. The answer to that lies in the fact that the issue in hand is group polarization, which is reasonably expected to strengthen with the strength of the discount of opposing beliefs. Thus, if some group polarization result is valid for a given discounting function, a group polarization result at least as strong should hold for discounting functions with less values.

We now discuss the properties of the upper bound function in \eqref{bound for general discounting function}, that is $1-(1-d)|x_i|^{\beta}$. First of all, for neutral agents, i.e., when $x_i \rightarrow 0$, it returns 1, which allows for the continuity with respect to $x_i$ of the broader $d_i$ characterized via \eqref{d 1} and \eqref{bound for general discounting function}. This is a very important property to satisfy if $d_i$ is to be realistic in any shape or form. Then, we focus on how $d$ and $\beta$, earlier branded as the degrees of freedom in the upper bound function, are interpreted. We first notice that $1-(1-d)|x_i|^{\beta}$ is non-decreasing in both $d$ and $\beta$. The parameter $d$ can be viewed as a uniform discount factor when $\beta$ is small. It also serves as an upper bound for the discount value employed by the extremely confident individuals, i.e., those with opinions close to 1 in absolute value. If $d=1$, the PC model converts to the PC-lite model. The parameter $\beta$ can be viewed as the discount's decay rate with respect to $|x_i|$. In other words, it captures the contribution of an individual's confidence to her discount value of opposing beliefs. The case where $\beta \rightarrow \infty$ turns the PC dynamics to its PC-lite counterpart.

\subsection{Properties of the PC Model}

    We now aim to investigate the behavior of the PC dynamics \eqref{ModifiedModelPart2}, with the discounting function $d_i$ characterized through \eqref{d 1} and \eqref{bound for general discounting function}. Special cases of \eqref{bound for general discounting function}, corresponding to marginal values of $d$ and $\beta$, are of particular interest to better understand the behavior of a general discounting function under conditions \eqref{d 1} and \eqref{bound for general discounting function}, and later a more general discounting function only restricted by \eqref{d 1}. As discussed earlier, either case of $d=1$ and $\beta \rightarrow \infty$ eliminates the confirmation bias and consequently simplifies the PC dynamics to the PC-lite dynamics, which was thoroughly investigated in Subsection~\ref{Properties of the First Model}. The marginal case $d=0$ will be later in Remark~\ref{d=0} argued to transpire the ``ultimate'' group polarization, where the opinions within a social bubble reach one of the very most extreme values $\pm 1$. The last marginal case, which will prove to be both interesting and informative, is that of $\beta \rightarrow 0$, that allows for an infinitely fast decay in the discount value with respect to $|x_i|$ and in limit amounts to a uniform upper bound on the discount value of opposing beliefs, i.e.,
    \begin{equation}
        d_i(x_i,x_j) \leq d ~\text{if}~\sgn(x_i)\neq\sgn(x_j).
        \label{d uniform}
    \end{equation}
In this case, Proposition~\ref{polarization_result} can be generalized as follows.

    \begin{proposition}
    \label{polarization_result_Second_model}
        Given the PC dynamics \eqref{ModifiedModelPart2}, with the discounting function $d_i$ satisfying \eqref{d 1} and \eqref{d uniform}, let a connected subset $\B\subseteq \V$ of agents have the bubble number
        \begin{equation}
        \label{low2}
            \gamma_\B>(3+2\sqrt{2})d,    
        \end{equation}
        and assume that $\alpha_1$ and $\alpha_2$, where $\alpha_1 < \alpha_2$, are the two positive solutions of the equation
        \begin{equation}
        \label{alphagamma2}
            \frac{1+\alpha}{\alpha(1-\alpha)} = \frac{\gamma_\B}{d}.
        \end{equation}
        If, for some $t_0$, it happens that
        \begin{equation}
        \label{alpha1_Second_model}
            x_i(t_0) > \alpha_1,~\forall i \in \B,
        \end{equation}
        then we have
        \begin{equation}
        \label{alpha2_Second_model}
            \liminf_{t \rightarrow \infty}{x_i(t)} \geq \alpha_2,~\forall i \in \B.
        \end{equation}
    \end{proposition}
    
    \begin{proof}
        Proposition~\ref{polarization_result_Second_model} will turn out to be a special case of Proposition~\ref{semi-general PC result}, stated later on in this section. See Subsection~\ref{rest of proofs} for more details.
    \end{proof}

    Just like Proposition~\ref{polarization_result}, Proposition~\ref{polarization_result_Second_model} also implies that if all agents within a social bubble reach a certain level of advocacy $+\alpha_1$ or disapproval $-\alpha_1$ of a position, as the interactions continue, agents in that relaxed bubble will reach a more extreme level of advocacy $+\alpha_2$ or disapproval $-\alpha_2$ of that position. Therefore, opinions are intensified and become more extreme or, equivalently, the opinions become polarized within that group. Consequently, a group polarization will occur in the direction of support/disapproval of a specific position. One also notices that as $d$ decreases, $\alpha_1$ and $\alpha_2$ will decrease and increase, respectively, as can be seen in Figure ~\ref{Alpha, model 2}. Therefore, if the conditions of Theorem 2 are satisfied in a social bubble, the result will be that agents with relatively low initial levels of advocacy/disapproval on a position will later have relatively extreme levels of advocacy/disapproval on that position.
    
% -------------------------------------------------   
    \begin{figure}[t]
        \centering
        \includegraphics[width=.5\linewidth]
        {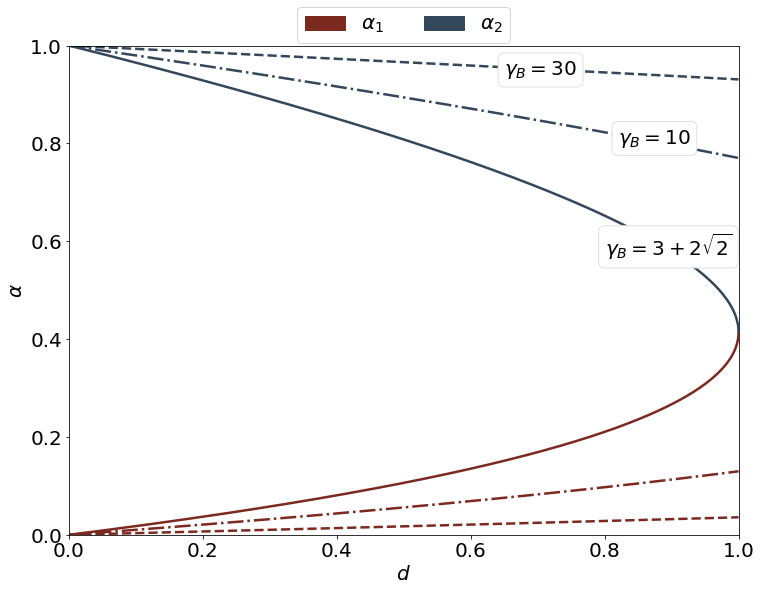}
        %corollary.PC.2.png
        \caption{Changes of $\alpha_1$ and $\alpha_2$ with respect to $\gamma_B$ and $d$.}
        \label{Alpha, model 2}
        %\label{AlphaChanges}
    \end{figure}
% ~~~~~~~~~~~~~~~~~~~~~~~~~~~~~~~~~~~~~~~~~~~~~~~~~

    For the purpose of completeness,  we generalize Proposition~\ref{polarization_result_Second_model} to the following proposition, which addresses group polarization under the PC dynamics for general $d$ and $\beta$.
    
    \begin{proposition}\label{semi-general PC result}
        Given the PC dynamics \eqref{ModifiedModelPart2}, with the discounting function $d_i$ satisfying \eqref{d 1} and \eqref{bound for general discounting function}, let a connected subset $\B\subseteq \V$ of agents have the bubble number $\gamma_\B$. Assume that equation
        \begin{equation}
        \label{alphagamma3}
            \frac{1+\alpha}{\alpha(1-\alpha)} = \frac{\gamma_\B}{1-(1-d)\alpha^{\beta}}
        \end{equation}
    has two positive solutions for $\alpha \in (0,1)$, namely $\alpha_1$ and $\alpha_2$, where $\alpha_1 \leq \alpha_2$. If, for some $t_0$, it happens that
        \begin{equation}
            x_i(t_0) > \alpha_1,~\forall i \in \B,
        \end{equation}
        then we have
        \begin{equation}
            \liminf_{t \rightarrow \infty}{x_i(t)} \geq \alpha_2,~\forall i \in \B.
        \end{equation}
    \end{proposition}
    
    \begin{proof}
        As it will be argued in Subsection~\ref{rest of proofs}, Proposition~\ref{semi-general PC result} is a special case of Theorem~\ref{general PC result}, stated later on in this section and proved in great detail in Subsection~\ref{main proof}. %ref-1
    \end{proof}
    
    The statement of Proposition~\ref{semi-general PC result} is different from those of Propositions~\ref{polarization_result} and \ref{polarization_result_Second_model} in that a succinct condition, such as \eqref{low1} and \eqref{low2}, under which \eqref{alphagamma3} is guaranteed to have solutions has not been provided. However, given any $d$ and $\beta$, it is straightforward to verify whether such solutions exist. We should also point out that a larger $\gamma_\B$, smaller $d$, and smaller $\beta$, all work in favor of \eqref{alphagamma3} having solutions for $\alpha$.
    
    \begin{remark}\label{d=0}
        In view of Proposition~\ref{semi-general PC result}, the marginal case $d=0$ can be interpreted to represent the ``ultimate'' group polarization. More precisely, given $d=0$, \eqref{alphagamma3} will have two positive solutions, $\alpha_1 < 1$ and $\alpha_2 = 1$, with the latter solution indicating the convergence of the opinions within the bubble to one of the very most extreme values $+1$ or $-1$. 
    \end{remark}
    
   The two degrees of freedom incorporated in the discounting function \eqref{discounting function- alpah and beta} can indicate sensitivity to an issue being discussed among the individuals. The higher the sensitivity to the issue for a specific population, the more intense the population acts in a biased manner towards it (smaller $d$ or $\beta$), the more probable/intense the polarization of opinions on that issue.% More precisely, some issues have higher potential to trigger and cause polarization among crowds. For instance, responds to political, religious and security issues, etc. might have higher potential to extremism with respect to other issues} %For instance,  mass shooting, sexual harassment and Muller investigation, have higher potential to extremism and polarization than white nationalism, North Korea and hurricanes, as declared by KF center \cite{KF-report-part1}. }
    
    Finally, Proposition~\ref{semi-general PC result} can be extended as follows to any discounting function restricted to conditions \eqref{d 1} and \eqref{relaxed d} for a non-increasing function $\hat{d}_i$.
    
    \begin{theorem}\label{general PC result}
        Given the PC dynamics \eqref{ModifiedModelPart2}, with the discounting function $d_i$ satisfying \eqref{d 1} as well as \eqref{relaxed d} for a non-increasing $\hat{d}_i:[0,1]\rightarrow [0,1]$, let a connected subset $\B\subseteq \V$ of agents have the bubble number $\gamma_\B$. Assume that inequality
        \begin{equation}
        \label{alphagamma4}
            \frac{1+\alpha}{\alpha(1-\alpha)} < \frac{\gamma_\B}{\hat{d}_i(\alpha)}
        \end{equation}
    is satisfied for any $i$ and $\alpha \in (\alpha_1,\alpha_2) \subseteq (0,1)$. If, for some $t_0$, it happens that
        \begin{equation}\label{ass alpha1}
            x_i(t_0) > \alpha_1,~\forall i \in \B,
        \end{equation}
        then we have
        \begin{equation}\label{st for alpha2}
            \liminf_{t \rightarrow \infty}{x_i(t)} \geq \alpha_2,~\forall i \in \B.
        \end{equation}
    \end{theorem}

\begin{proof}
    The proof of Theorem~\ref{general PC result} is given in Subsection~\ref{main proof}.%ref-2
\end{proof}

\section{Numerical Experiments}
\label{Numerical Examples}
    In this section, we illustrate the behaviors of PC-lite dynamics \eqref{model} and PC dynamics \eqref{ModifiedModelPart2} through numerical examples. In all examples, a fixed Erd\"os–R\'enyi random graph \cite{erdHos1960evolution} embodies the underlying graph of the network. It consists of $|\V|=500$ nodes and, for each pair of nodes $i,j \in \V$, an edge $e_{ij}$ exists with the uniform probability $p_{G}=0.06$, independently of other edges. Each edge is then independently activated at any time step with the uniform probability $p_L = 0.8$, which results in asynchronous interactions in the network. Numerical examples of the PC-lite model and the PC model are provided in subsections~\ref{Numerical Examples of model 1} and \ref{Numerical Examples of model 2} below, respectively.
% -------------------------------------------------
\subsection{Numerical Examples for the PC-lite Model}
\label{Numerical Examples of model 1}
    For the numerical examples of PC-lite dynamics \eqref{model}, we consider two cases, (i) a case where there are three bubbles within the population, while some agents do not belong to any of these bubbles, and (ii) a case where the entire population is divided into two bubbles. Other parameters used in the simulations of these two cases, including the size of the bubbles, their bubble numbers, and their respective $\alpha_1$ and $\alpha_2$ values, are given in Tables~\ref{table 1} and \ref{table 2}. For both cases, the initial opinions of the agents within each bubble are selected according to a normal probability distribution function with near-zero mean and low variance, as specified in Tables~\ref{table 1} and \ref{table 2} with $\mu$ and $\sigma^2$ representing the corresponding mean and variance, truncated to the range $[-1,1]$. The initial opinions of the agents outside the bubbles in the first case are selected according to a normal probability distribution function with zero mean and variance equal to 0.11. Finally, we note that the weight values at any time step are generated randomly but scaled in such a way not to violate the bubble numbers listed in Tables~\ref{table 1} and \ref{table 2}. More precisely, given an agent inside a bubble, the weights indicating the influence over her from outside are uniformly scaled down. 
%%%%%%%%%%%%%%%%%%%%%%%%%%
    \begin{table}[h!]
    \centering
    \caption{Parameters of the PC-lite model simulation (Figure~\ref{PC-lite A})}
    \begin{tabular}{cccc}
     \hline
     Parameter &  Bubble 1 & Bubble 2 & Bubble 3 \\
     \hline
     $|\mathcal{B}|$      &   159  & 127 & 90\\ 
     $\gamma_\mathcal{B}$ & 8  & 6  & 12 \\ 
     $\alpha_1$           & 0.18  &  0.333  & 0.102 \\
     $\alpha_2$           & 0.695  &  0.5  &  0.814\\
     $\mu$                & 0.04 & -0.02 & -0.01 \\
     $\sigma^2$           & 0.1 & 0.09 & 0.08\\ 
    \hline
    \end{tabular} 
    \label{table 1}
    \end{table}
%%%%%%%%%%%%%%%%%%%%%%%%%%
    \begin{table}[h!]
    \centering
    \caption{Parameters of the PC-lite model simulation (Figure~\ref{PC-lite B1})}
    \begin{tabular}{ccc}
     \hline
     Parameters &  Bubble 1 & Bubble 2 \\
     \hline
     $|\mathcal{B}|$      &   261  & 239\\ 
     $\gamma_\mathcal{B}$ & 9  & 6\\ 
     $\alpha_1$           & 0.15  &  0.333 \\
     $\alpha_2$           & 0.738  &  0.5\\
     $\mu$                 &  0.06  &    -0.08   \\
     $\sigma^2$            &  0.25  &   0.3 \\
    \hline
    \end{tabular} 
    \label{table 2}
    \end{table}

%%%%%
    Figure~\ref{PC-lite A} demonstrates the evolution of opinions under the PC-lite dynamics for the first case. It confirms the statement of Proposition~\ref{polarization_result} that if the opinions of the agents in Bubble~1 become greater than $\alpha_1$ in finite time, they will in the long run become more extreme than $\alpha_2$. The same conclusion can be drawn for Bubble~2 and Bubble~3. Furthermore, Bubble~3 appears to show a stronger group polarization than Bubble~1 and Bubble~2, which is consistent with it having a larger $\alpha_2$ than the other bubbles, that in view of Figure~\ref{Alpha, model 1} is a result of its relatively large bubble number.% In addition, since there are three bubbles $m=3$, according to Remark~\ref{corollary 1}, the network can exhibit at least $2^m=8$ substantially different limiting behaviors.

\begin{figure}[t]
\centering
\includegraphics[width=0.5\linewidth]
{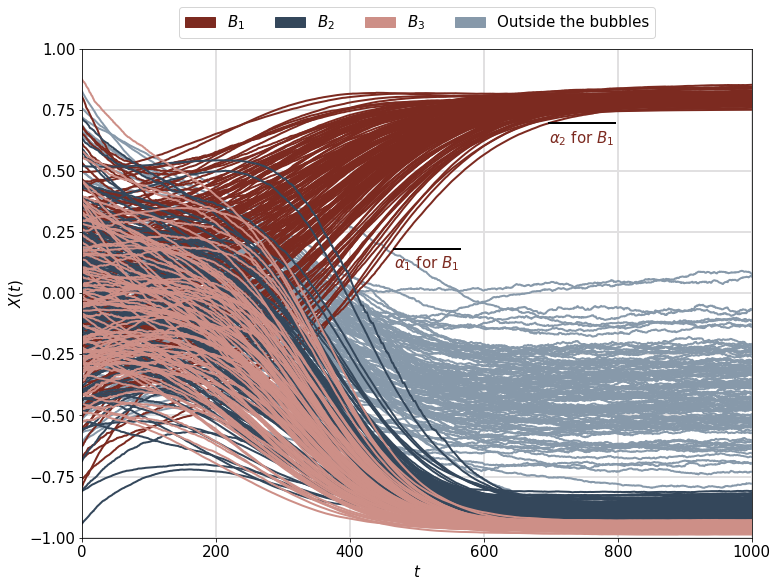}
\caption{Opinion evolution according to PC-lite dynamics \eqref{model} with parameters specified in Table~\ref{table 1}.% $p_G=0.06$, $p_L=0.8$. For agents who do not belong to any bubbles $|V|-\left(|B_1|+|B_2|+|B_3|\right)=124$, $\mu=0$, and $\sigma^2=0.11$, and other parameters are given in Table~\ref{table 1}. Time of tipping points, $t_0$'s as elaborated in Proposition~\ref{polarization_result}, are greater than zero.
}
\label{PC-lite A}
\end{figure}   
%%%%%%%%%%%%%%%%%%%%%%%%%%
    
    Figure~\ref{PC-lite B1} shows the evolution of opinions under the PC-lite dynamics for the second case. While it confirms the statement of Proposition~\ref{polarization_result} like the previous simulation, it is designed to resemble the opinion polarization of the US population depicted in Figure~\ref{US polarization}. More specifically, it shows how the medians of the opinions in the two bubbles diverge over time. The distribution of opinions at time steps 0 and 200 are separately drawn in Figure~\ref{PC-lite B3}, bearing a resemblance to Figure~\ref{US polarization}.
    % and \ref{PC-lite B3}
    % According to \cite{erdHos1960evolution}, the threshold of probability for connectivity of the generated graph $\mathcal{G}(n,p)$ is $\frac{ln(n)}{n}$. Thus for $p > \frac{(1+\epsilon)ln(n)}{n} $ the generated graph is connected and for $p < \frac{(1-\epsilon)ln(n)}{n} $ the generated graph is not connected, where $\epsilon \in \mathbb{R}^{>0}$ is a small number.

% -------------------------------------------------

\begin{figure}[t]
\centering
\includegraphics[width=0.5\linewidth]
{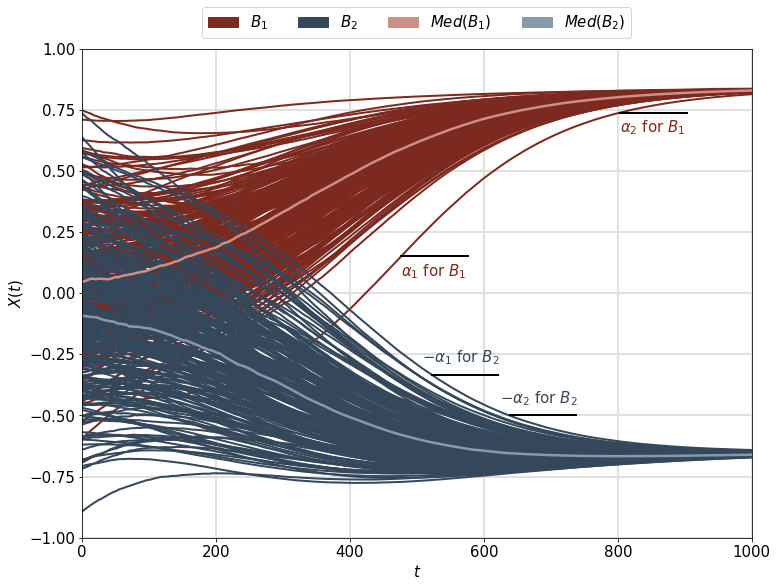} 
\caption{Opinion evolution according to PC-lite dynamics \eqref{model} with parameters specified in Table~\ref{table 2}.% with bubbles. $p_G=0.05$, $p_L=0.8$, and other parameters are given in Table~\ref{table 1}. Initial condition of each bubble are sampled randomly from to normal probability distribution. Time of tipping points, $t_0$'s as elaborated in Proposition~\ref{polarization_result}, are greater than zero.
}
\label{PC-lite B1}
\end{figure}
%%%%%%%%%%%%%%%%%%%%%%%%%%%%%%%(B3)###############################

    \begin{figure}[t]
    \centering 
    \begin{subfigure}[h]{0.4\textwidth}
    \centering 
    \includegraphics[width=0.95\linewidth]
    {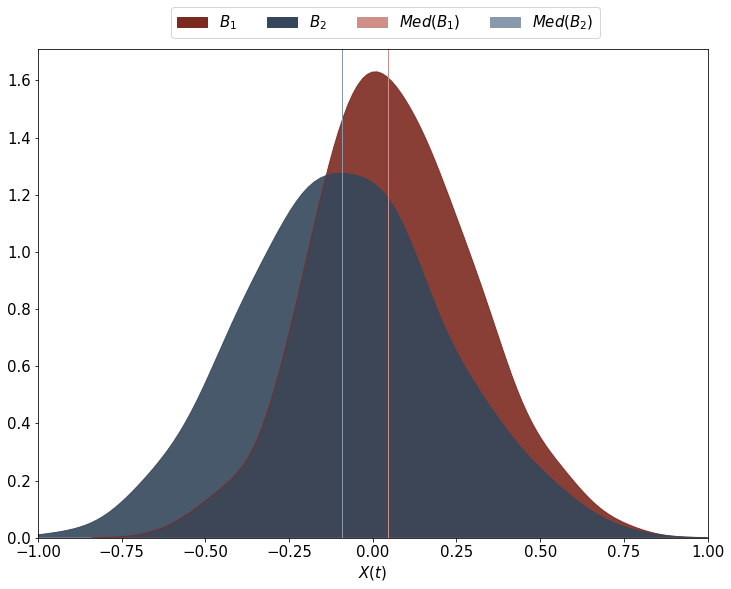} %model(1)_2_v1.0.a.png
    \label{normal snapshot a}
    \caption{}
    \end{subfigure}
    %
    %\quad
    \hspace{10px}
    \begin{subfigure}[h]{0.4\textwidth}
    \centering 
    \includegraphics[width=0.95\linewidth]
    {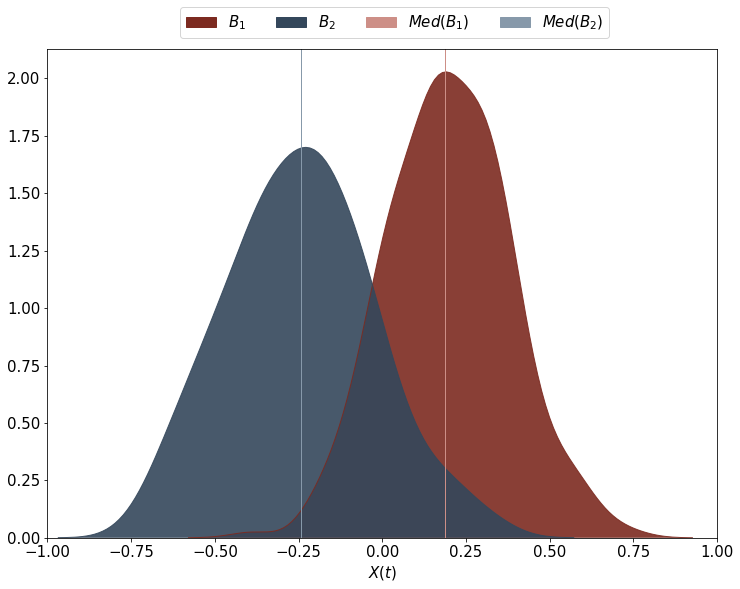}
    %model(1)_2_v1.0.b.png
    \label{normal snapshot b}
    \caption{}
    \end{subfigure}
    %
    %\quad
    \hspace{10px}
    %
    % \begin{subfigure}[h]{0.4\textwidth}
    % \centering 
    % \includegraphics[width=\linewidth]
    % {Fig.6.1-c.png}
    % %model(1)_2_v1.0.c.png
    % \label{normal snapshot c}
    % \caption{}
    % \end{subfigure}
    
    \caption{Opinion distribution in each bubble appearing in Figure~\ref{PC-lite B1} at
    (a) $t=0$ and (b) $t=200$.}
    \label{PC-lite B3}
    \end{figure}

% -------------------------------------------------
\subsection{Numerical Examples for the PC Model}
\label{Numerical Examples of model 2}
    %Here, we simulate the dynamics of the ubiquity model \eqref{ModifiedModelPart2}. To fit that model to a real-world network with the small-world property \footnote{The small-world property states that a real{\color{blue}-world} network {\color{blue}(graph)} has a high clustering coefficient and a small average shortest path length.}, we generated a fixed underlying graph is based on a Watts–Strogatz random graph model \cite{watts1998collective}. According to \cite{watts1998collective}, by randomly rewiring edges of a $k$-regular graph, with a uniform probability $p_{G}$ within the interval $p_{G} \in [0.01,0.1]$, one can reach a spectrum of random graphs that have the small-world property. 
    For the numerical examples of PC dynamics \eqref{ModifiedModelPart2}, we consider a case where the network is divided into two bubbles, the parameters of which are given in Table~\ref{table 3}. Figure~\ref{PC fig} demonstrates the evolution of opinions under the PC dynamics \eqref{ModifiedModelPart2} where
    \begin{equation}
        d_i(x_i,x_j) = \begin{cases}
            1 & \text{if~}\sgn(x_i)= \sgn(x_j)\\
            1-(1-d)|x_i|^{\beta}& \text{if~}\sgn(x_i)\neq \sgn(x_j)
        \end{cases}
    \end{equation} with $d=0.4$ and $\beta=0.4$. One can observe in Figure~\ref{PC fig} that all the opinions of agents in the bubbles will asymptotically exceed their respective $\alpha_2$'s in magnitude, confirming Proposition~\ref{polarization_result_Second_model}.
    
% peaking near $\alpha_1$, resembling a weak in-bubble confidence in the beginning. 

% For each bubble, the initial opinions are selected in such a way that the conditions of Proposition~\ref{polarization_result_Second_model} are satisfied at $t_0=0$ for some values $\alpha_1$ and $\alpha_2$.

    \begin{table}[h!]
        \centering
        \caption{Parameters of the PC model simulation (Figure \ref{PC fig})}
    \begin{tabular}{ccc}
         \hline
         Parameters & Bubble 1 & Bubble 2 \\
         \hline
         $|\mathcal{B}|$ & 226 & 274\\ 
         $\gamma_\mathcal{B}$ & 3.5 & 4\\ 
         $\alpha_1$ & 0.3721  &  0.2869\\
         $\alpha_2$ &  0.6287 &  0.7149 \\
         $\mu$      &   0.09     &  -0.08 \\
         $\sigma^2$ &  0.1      &   0.11\\
         \hline
    \end{tabular} 
    \label{table 3}
    \end{table}

    \begin{figure}[ht!]
     \centering
      \includegraphics[width=0.5\linewidth]
        {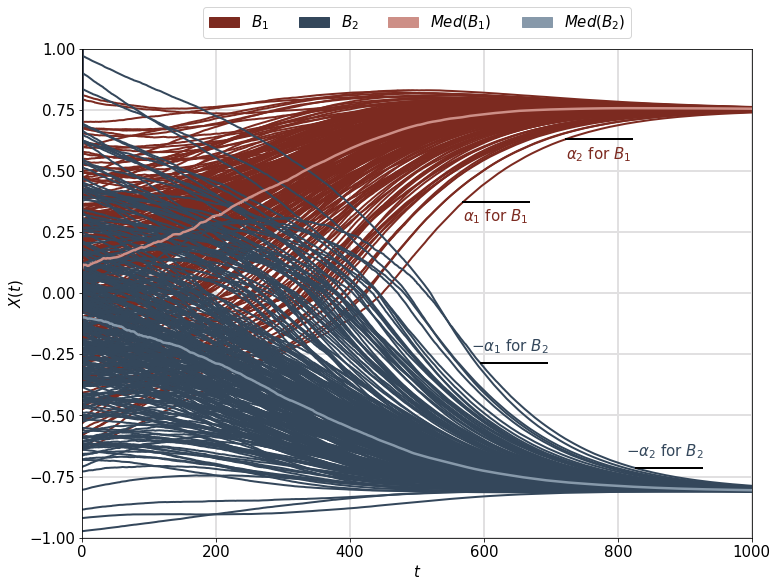}
        \caption{Opinion evolution according to PC dynamics \eqref{ModifiedModelPart2} with parameters specified in Table~\ref{table 3}.}
    \label{PC fig}%for both (a) and (b) 
    \end{figure}
% ------------------------------------------------- 

    %and $d=0.2$ correspond to sub figures \ref{d=0.1} and \ref{d=0.2} respectively. In Figure \ref{model 2. sim}, because of the clustered nature of the underlying graph, bubbles are formed over time, and the agents in each bubble went toward grater risk.

% -------------------------------------------------

% ~~~~~~~~~~~~~~~~~~~~~~~~~~~~~~~~~~~~~~~~~~~~~~~~~
\section{Conclusion}
\label{Conclusion}

    In this paper, we proposed the PC-lite and PC models of opinion dynamics based on an approach that views an opinion via its intensity and its direction. We established a result on the occurrence of opinion polarization in a social bubble, referring to a group of individuals who are highly cohesive and isolated from outside influence. Both of the models developed are inspired by the notion of groupthink, widely studied in the socio-psychological literature. We also justified our results using numerical simulations.
    
    The ultimate goal of the proposed models is to analyze, predict, and possibly intervene in the process of group polarization and opinion polarization in a population. While the analysis and prediction goals were discussed in this work, the intervention techniques will remain as part of future work. As another future research direction, it will be interesting to study the multidimensional version of the proposed models that captures the simultaneous evolution of opinions on a multitude of correlated topics.
% ~~~~~~~~~~~~~~~~~~~~~~~~~~~~~~~~~~~~~~~~~~~~~~~~~
\section{Proofs}\label{Proofs}

    This section is composed of a subsection containing some preliminary definitions and lemmas that is integral in the proof of Theorem~\ref{general PC result}, an entire subsection detailing the proof of Theorem~\ref{general PC result}, and another subsection on the derivations of Propositions~\ref{prop:polarization_result}, \ref{polarization_result}, \ref{polarization_result_Second_model}, and \ref{semi-general PC result} from Theorem~\ref{general PC result}.
    
    \subsection{Preliminaries to the Proof of Theorem~\ref{general PC result}} \label{preliminary proof} %label-1
    
    We recall that the discounting function in Theorem~\ref{general PC result} is assumed to satisfy \eqref{d 1} and \eqref{relaxed d} for non-increasing functions $\hat{d}_i$. Since the marginal case of \eqref{relaxed d} will prove to be of great importance, we define an auxiliary discounting function $d':[-1,1]^2 \rightarrow [0,1]$ by
    \begin{equation}
        d'_i(x_i,x_j) = 1 ~ \text{if}~ \sgn(x_i)=\sgn(x_j),
    \label{d' 1}
    \end{equation}
    \begin{equation}
        d'_i(x_i,x_j) = \hat{d}_i(|x_i|)~\text{if}~\sgn(x_i)\neq\sgn(x_j),
        \label{relaxed d equality}
    \end{equation}
    where \eqref{d' 1} is identical to \eqref{d 1}, while \eqref{relaxed d equality} is the marginal case of \eqref{relaxed d}. For any $y \in [-1,1]^n$, it should be clear that
    \begin{equation}
        d_i(y_i,y_j) \leq d'_i(y_i,y_j),~\forall i,j.
    \label{dd'}
    \end{equation}
    We also define a function $f:[-1,1]^n \rightarrow [-1,1]^n$ with its $i$th coordinate formulated as
        \begin{equation}
            f_i(y)\triangleq y_i+\sum_{j \neq i} \Big[ d'_i(y_i,y_j) w_{ij}|y_j|\left( \sgn{(y_j)}-y_i \right) \Big],
            \label{f def}
        \end{equation}
    There is a slight abuse of notation in \eqref{f def}, in that $w_{ij}$ is in general a function of time, while $f$ does not seem to be treated as one. This will not cause a problem since the time index will be fixed whenever $f_i$ will show up in future arguments. In particular, $f_i(x(t))$ can be expressed as
    \begin{equation}
        f_i(x(t))= x_i(t) + \sum_{j \neq i} \Big[d'_i(x_i(t),x_j(t)) w_{ij}(t) |x_j(t)| (\sgn(x_j(t)) - x_i(t))\Big].
    \label{f dyn}
    \end{equation}
    In contrast, according to \eqref{ModifiedModelPart2},
    \begin{equation}
        x_i(t+1) = x_i(t) + \sum_{j \neq i} \Big[d_i(x_i(t),x_j(t)) w_{ij}(t) |x_j(t)| (\sgn(x_j(t)) - x_i(t))\Big].
    \label{pc alt}
    \end{equation}
    The following lemmas will be used in the proof of Theorem~\ref{general PC result}.
    
    \begin{lemma}\label{f dyn pc alt}
        For an arbitrary agent $i$, if $x_i(t) \geq 0$, then $x_i(t+1) \geq f_i(x(t))$.
    \end{lemma}
    
    \begin{proof}
        From \eqref{f dyn} and \eqref{pc alt},
        \begin{equation}
            x_i(t+1) - f_i(x(t)) = \sum_{j \neq i} \Big[\big(d_i(x_i(t),x_j(t))-d'_i(x_i(t),x_j(t))\big) w_{ij}(t) |x_j(t)| (\sgn(x_j(t)) - x_i(t))\Big].
        \label{diff}
        \end{equation}
    To complete the proof, it is sufficient to show that each summand (term appearing in a summation) in \eqref{diff} is non-negative, which can be done simply by considering the two cases for $\sgn(x_j(t))$. If $\sgn(x_j(t)) = +1$, given the assumption $x_i(t) > 0$, the summand becomes zero since both $d_i$ and $d'_i$ would equal 1. If $\sgn(x_j(t)) = +1$, then from \eqref{dd'}, the summand is non-negative. 
    \end{proof}

    \begin{lemma}\label{increasing_model}
        The function $f$ is non-decreasing, i.e., for any pair of vectors $y^1,y^2 \in [-1,1]^n$,
        \begin{equation}
        \label{inc}
            y^1 \leq y^2 \Rightarrow f(y^1) \leq f(y^2),
        \end{equation}
        where the inequalities in \eqref{inc} are to be understood element-wise.% In particular, the PC-lite dynamics \eqref{model}, derived from the PC dynamics with $\hat{d}_i$ as a constant function with value $1$, is non-increasing.
    \end{lemma}
\begin{proof}
Without loss of generality, we can assume that $y^1$ and $y^2$ differ only in one coordinate, say the $k$th coordinate. For notational convenience, we may write $y_i$ for both $y^1_i$ and $y^2_i$, when $i \neq k$. We show that $f_i(y^1) \leq f_i(y^2)$ for each $i$ by considering the following cases:

{\bf Case 1:} $(i \neq k)$ In this case,
    \begin{align}\label{case 1}
        f_i(y^2) - f_i(y^1)
        & =\left(y_i + \sum_{j\neq i} \Big[d_i(y_i,y^2_j)w_{ij} |y^2_j|(\sgn(y^2_j) - y_i)\Big]\right)\cr
        & -\left( y_i + \sum_{j \neq i} \Big[ d_i(y_i,y^1_j)w_{ij} |y^1_j|(\sgn(y^1_j) - y_i)\Big]\right)\cr
        & =w_{ik}\Big(d_i(y_i,y^2_k)|y^2_k|(\sgn(y^2_k) - y_i)- d_i(y_i,y^1_k)|y^1_k|(\sgn(y^1_k) - y_i)\Big)
    \end{align}
Now, if $\sgn(y^1_k) = \sgn(y^2_k)$, then from \eqref{d 1} and \eqref{relaxed d equality} we have $d_i(y_i,y^1_k) = d_i(y_i,y^2_k)$, which combined with \eqref{case 1} leads to
        \begin{align}\label{case 1 cont}
        f_i(y^2) - f_i(y^1)
        & =w_{ik} d_i(y_i,y^1_k)\Big(|y^2_k|(\sgn(y^2_k) - y_i)- |y^1_k|(\sgn(y^1_k) - y_i)\Big)\cr
        & =w_{ik}d_i(y_i,y^1_k) \Big( (y^2_k - y^1_k) - y_i (|y^2_k|-|y^1_k|) \Big)\cr
        & \geq w_{ik}d_i(y_i,y^1_k)(1-|y_i|)(y^2_k-y^1_k)\cr
        & \geq 0,
    \end{align}
where in the first inequality of \eqref{case 1 cont}, we used the inequality $\big||y^2_k|-|y^1_k|\big|\leq |y^2_k-y^1_k|=y^2_k-y^1_k$. On the other hand, if $\sgn(y^1_k) \neq \sgn(y^2_k)$, given $y^2_k \geq y^1_k$, we have $\sgn(y^1_k) = -1$ and $\sgn(y^2_k) = 1$, which considered together with \eqref{case 1}, immediately results in $f_i(y^2) - f_i(y^1) \geq 0$.

{\bf Case 2:} $(i = k)$ In this case,
    \begin{align}\label{lem f}
        f_i(y^2) - f_i(y^1)
        & = \left(y^2_k + \sum_{j \neq k} \Big[ d_k(y^2_k,y_j)w_{kj} |y_j|(\sgn(y_j) - y^2_k)\Big] \right)\cr
        & -\left( y^1_k + \sum_{j \neq k} \Big[ d_k(y^2_k,y_j)w_{kj} |y_j|(\sgn(y_j) - y^1_k)\Big] \right)\cr   
        & =\left(y^2_k - y^1_k\right) + \sum_{j \neq k} \Big[ w_{kj}|y_j|\Big( d_k(y^2_k,y_j) (\sgn(y_j) - y^2_k) - d_k(y^2_k,y_j) (\sgn(y_j) - y^1_k) \Big)\Big]\cr
        & = \left(y^2_k - y^1_k\right) + \sum_{j \neq k} \Big[ w_{kj}|y_j|\Big(d_k(y^1_k,y_j)(y^1_k - y^2_k) + \big(d_k(y^2_k,y_j)-d_k(y^1_k,y_j)\big)(\sgn(y_j)-y^2_k)\Big)\Big].\cr
        &~
    \end{align}
Considering the two cases $\pm 1$ for $\sgn(y_j)$, and remembering $y^2_k \geq y^1_k$ as well as \eqref{d 1} and \eqref{relaxed d equality}, the term
\begin{equation}
    \big(d_k(y^2_k,y_j)-d_k(y^1_k,y_j)\big)(\sgn(y_j)-y^2_k)
\end{equation}
which appears in the last line of \eqref{lem f}, can be easily shown to be zero or positive. Hence, \eqref{lem f} results in
\begin{align}\label{lem f2}
        f_i(y^2) - f_i(y^1)
        & \geq \left(y^2_k - y^1_k\right) + \sum_{j \neq k} \Big[ w_{kj}|y_j|\Big(d_k(y^1_k,y_j)(y^1_k - y^2_k) \Big)\Big]\cr
        & = \left(y^2_k - y^1_k\right) \left( 1 - \sum_{j \neq k} \Big[ w_{kj}|y_j| d_k(y^1_k,y_j) \Big] \right)\cr
        & \geq \left(y^2_k - y^1_k\right) \left( 1 - \sum_{j \neq k}w_{kj}\right)\cr
        &\geq 0.
\end{align}
\end{proof}

% \begin{proof}
% Without loss of generality, we can assume that $y^1$ and $y^2$ differ only in one coordinate, say the $k$th coordinate. For notational convenience, we may write $y_i$ for both $y^1_i$ and $y^2_i$, when $i \neq k$. We show that $f_i(y^1) \leq f_i(y^2)$ for each $i$ by considering the following cases:

% {\bf Case 1:} $(i \neq k)$ In this case,
%     \begin{align}
%         f_i(y^2) - f_i(y^1)
%         & =\left(y^2_i + \sum_j w_{ij} |y^2_j|(\sgn(y^2_j) - y^2_i)\right)\cr
%         & -\left( y^1_i + \sum_j w_{ij} |y^1_j|(\sgn(y^1_j) - y^1_i)\right)\cr
%         & =w_{ik}\left(|y^2_k|(\sgn(y^2_k) - y_i)- |y^1_k|(\sgn(y^1_k) - y_i)\right)\cr
%         & =w_{ik} \left( (y^2_k - y^1_k) - y_i (|y^2_k|-|y^1_k|) \right)\cr
%         & \geq w_{ik}(1-|y_i|)(y^2_k-y^1_k)\cr
%         & \geq 0.
%     \end{align}

% {\bf Case 2:} $(i = k)$ In this case,
%     \begin{align}
%         f_k(y^2) - f_k(y^1)
%         & = \left(y^2_k + \sum_j w_{ij} |y_j|(\sgn(y_j) - y^2_k)\right)\cr
%         & -\left( y^1_k + \sum_j w_{ij} |y_j|(\sgn(y_j) - y^1_k)\right)\cr   
%         & =(y^2_k - y^1_k)\left(1 - \sum_j w_{ij}|y_j|\right)\cr
%         & \geq (y^2_k - y^1_k)\left(1 - \sum_j w_{ij}\right)\cr
%         & \geq 0.
%     \end{align}
% \end{proof}

% ~~~~~~~~~~~~~~~~~~~~~~~~~~~~~~~~~~~~~~~~~~~~~~~~~

\subsection{Proof of Theorem~\ref{general PC result}} \label{main proof} %label-2

    Defining $z(t) = \min_{i \in \B} x_i(t)$, the inequality \eqref{st for alpha2} holds for each $i\in\B$ if and only if
    \begin{equation}\label{pp1eq}
        \liminf_{t \rightarrow\infty}{z(t)} \geq \alpha_2.
    \end{equation}
    
    To show \eqref{pp1eq}, we first note that, by the assumption of the theorem, that is \eqref{ass alpha1}, we have
    \begin{equation}\label{z1}
        z(t_0) > \alpha_1.
    \end{equation}
    Given an arbitrary but fixed $t$, we then take the following six steps to fully examine $\liminf_{t \rightarrow\infty}{z(t)}$ and show \eqref{pp1eq}.
    
    {\bf Step 1:} We show that the following statement is true:
    \begin{equation}\label{z2}
        \Big(\alpha_1 < z(t) < \alpha_2\Big) ~~\Longrightarrow~~ \Big(z(t+1) \geq z(t)\Big).
    \end{equation}
    To this aim, construct a vector $y$ from $x(t)$ as
    \begin{equation}
        y_i =
        \begin{cases}
            z(t)    & \text{if } i \in \B\\
            -1  & \text{if } i \not\in \B.
        \end{cases}
    \end{equation}
    It should be clear that $x(t) \geq y$. Thus, from Lemma~\ref{increasing_model}, we must have $f(x(t)) \geq f(y)$, and in particular, $f_i(x(t)) \geq f_i(y)$, $\forall i \in \B$. On the other hand, since $x_i(t) \geq 0$ for $i \in \B$, from Lemma~\ref{f dyn pc alt}, we conclue that $x_i(t+1) \geq f_i(x(t))$. Hence, $x_i(t+1) \geq f_i(y)$. Therefore, for each $i \in \B$,
    \begin{align}\label{arg1}
        x_i(t+1)
        & \geq f_i(y)\cr
        & = y_i + \sum_{j \in \V} d'_i(y_i,y_j)w_{ij}(t) |y_j| (\sgn(y_j) - y_i) \cr
        & = y_i + \sum_{j \in \B} d'_i(y_i,y_j)w_{ij}(t) |y_j| (\sgn(y_j) - y_i) \cr
        & + \sum_{j \in \V \backslash \B} d'_i(y_i,y_j)w_{ij}(t) |y_j| (\sgn(y_j) - y_i)\cr
        & = z(t) + \sum_{j \in \B} w_{ij}(t) z(t)(1- z(t)) \cr
        & + \sum_{j \in \V \backslash \B} \hat{d}_i(z(t))w_{ij}(t) (-1-z(t))\cr
        & = z(t) + z(t)(1-z(t))\cr
        &\times\left[ \sum_{j \in \B} w_{ij}(t) - \frac{\hat{d}_i(z(t))(1+z(t))}{z(t)(1-z(t))}\sum_{j \in \V \backslash \B} w_{ij}(t)\right].
    \end{align}    
    Furthermore, since $\alpha_1 < z(t) < \alpha_2$, from the assumption \eqref{alphagamma4} we have
    \begin{equation}
    \label{zgammaB}
        \frac{\hat{d}_i(z(t))(1+z(t))}{z(t)(1-z(t))} < \gamma_\B.
    \end{equation}
    Combining \eqref{arg1} and \eqref{zgammaB} implies that
    \begin{align}
    \label{arg1'}
        x_i(t+1)
        & \geq z(t) + z(t)(1-z(t))\left[ \sum_{j \in \B} w_{ij}(t) - \gamma_\B \sum_{j \in \V \backslash \B} w_{ij}(t)\right]\cr
        & \geq z(t).
    \end{align}
    Consequently, $z(t+1) \geq z(t)$, which completes the proof of statement \eqref{z2}.
    
    {\bf Step 2:} We show that
    \begin{equation}\label{z3}
        \Big(z(t) \geq \alpha_2\Big) ~\Rightarrow~ \Big(z(t+1) \geq \alpha_2\Big).
    \end{equation}
    To prove statement \eqref{z3}, we first point out that since \eqref{alphagamma4} is assumed to hold for any $i$ and $\alpha \in (\alpha_1,\alpha_2)
    $, one can write
    \begin{equation}\label{alpha2magic}
        \hat{d}_i(\alpha_2) \leq \lim_{\alpha \rightarrow \alpha_2}\hat{d}_i(\alpha) \leq \lim_{\alpha \rightarrow \alpha_2}\frac{\alpha_2(1-\alpha_2)\gamma_\B}{1+\alpha_2} = \frac{\alpha_2(1-\alpha_2)\gamma_\B}{1+\alpha_2},
    \end{equation}
where the first inequality of \eqref{alpha2magic}, as well as the existence of $\lim_{\alpha \rightarrow \alpha_2}\hat{d}_i(\alpha)$, are both immediate results of $\hat{d}_i$ being non-increasing, while the second inequality of \eqref{alpha2magic} is a direct consequence of \eqref{alphagamma4}. Rewriting \eqref{alpha2magic}, we have
    \begin{equation}\label{alpha2magic2}
        \frac{1+\alpha_2}{\alpha_2(1-\alpha_2)} \leq \frac{\gamma_\B}{\hat{d}_i(\alpha_2)}
    \end{equation}
    Now, we follow the same line of arguments as in Step 1, only here we construct the vector $y$ as
    \begin{equation}
        y_i =
        \begin{cases}
            \alpha_2    & \text{if } i \in \B\\
            -1  & \text{if } i \not\in \B,
        \end{cases}
    \end{equation}
    and replace $z(t)$ in \eqref{arg1} by $\alpha_2$. Then, we write \eqref{zgammaB} for $\alpha_2$ in place of $z(t)$ by employing \eqref{alpha2magic2} in place of  \eqref{alphagamma4}. These modifications of \eqref{arg1} and \eqref{zgammaB} together imply $x_i(t+1) \geq \alpha_2$, $\forall i \in \B$, which proves statement \eqref{z3}.
    
    {\bf Step 3:} Combining Steps~1 and 2, from \eqref{z1}, \eqref{z2} and \eqref{z3}, we conclude that \eqref{pp1eq} holds unless $z(t)$ is non-decreasing for every $t \geq t_0$ and $\lim_{t\rightarrow\infty} z(t)$ exists and lies in the interval $(\alpha_1,\alpha_2)$. Thus, assume on the contrary that $z(t)$ is non-decreasing for $t \geq t_0$ and
    \begin{align}
    \label{zlim}
        \lim_{t\rightarrow\infty} z(t)=z^*\in (\alpha_1,\alpha_2).
    \end{align}
    Let $\epsilon > 0$ be sufficiently small that it satisfies
    \begin{align}
    \label{small espsilon}
        \gamma_\B > \frac{\hat{d}_i(z^*-\epsilon)(1+z^*+\epsilon)}{(z^*-\epsilon)(1-(z^*+\epsilon))},
    \end{align}
    for any $i \in \V$. One notices that \eqref{small espsilon} holds for any sufficiently small $\epsilon$ since, as $\epsilon$ vanishes, the right-hand expression of \eqref{small espsilon} converges to $\hat{d}_i(z^*)(1+z^*)/[z^*(1-z^*)]$, which is less than $\gamma_\B$. According to \eqref{zlim}, there exists $T > t_0$ such that
    \begin{align}
        z(t) > z^*-\epsilon,~\forall t > T.
    \end{align}
    
    {\bf Step 4:} Let $i \in \B$ be arbitrary. We show that there is a time instant $t_i > T$ such that
    \begin{align}
    \label{t_i}
        x_i(t_i) > z^*+\epsilon.
    \end{align}
    Assume to the contrary that $x_i(t)$ never exceeds $z^*+\epsilon$ after time $T$, that is $x_i(t) \leq z^* + \epsilon$ for any $t > T$. Now, for any $t > T$, given the PC dynamics we have
    \begin{align}
    \label{znew4}
        x_i(t+1)-x_i(t)
        & = \sum_{j \in \V} d_i(x_i(t),x_j(t))w_{ij}(t)|x_j(t)|(\sgn(x_j(t))-x_i(t))\cr
        & \geq \left(\sum_{j \in \B} w_{ij}(t)\right)(z^*-\epsilon)(1-(z^*+\epsilon))\cr
        &+ \hat{d}_i(z^*-\epsilon)\left(\sum_{j \in \V\backslash\B} w_{ij}(t)\right)(-1-(z^*+\epsilon))\cr
        &\geq \left(\sum_{j \in \B}w_{ij}(t)\right) \left[ (z^*-\epsilon)(1-(z^*+\epsilon)) + \frac{\hat{d}_i(z^*-\epsilon)}{\gamma_\B} (-1-(z^*+\epsilon)) \right].
    \end{align}
    Summing up \eqref{znew4} over consecutive time instants, we conclude that
    \begin{align}
    \label{znew4'}
        x_i(t')-x_i(t) \geq \left(\sum_{\tau=t}^{t'-1}\sum_{j \in \B}w_{ij}(\tau)\right) \left[ (z^*-\epsilon)(1-(z^*+\epsilon)) + \frac{\hat{d}_i(z^*-\epsilon)}{\gamma_\B} (-1-(z^*+\epsilon)) \right].
    \end{align}
    The right-hand expression in \eqref{znew4'} explodes as $t'$ grows since
    \begin{align}
        \left[ (z^*-\epsilon)(1-(z^*+\epsilon)) + \frac{\hat{d}_i(z^*-\epsilon)}{\gamma_\B} (-1-(z^*+\epsilon)) \right]
    \end{align}
    is lower-bounded by a positive number according to \eqref{small espsilon} and
    \begin{align}
        \sum_{\tau=t}^{t'-1}\sum_{j \in \B}w_{ij}(\tau)
    \end{align}
    grows unbounded as $t'\rightarrow \infty$ since $\B$ is connected. This is a contradiction, meaning that there is $t_i > T$ for which \eqref{t_i} holds.
    
    {\bf Step 5:} Let $i \in \B$ be arbitrary. We show that if $t > T$,
    \begin{equation}
    \label{znew1}
        \Big(x_i(t) \geq z^*+\epsilon\Big) ~~\Longrightarrow~~ \Big(x_i(t+1) \geq z^*+\epsilon\Big).
    \end{equation}
    According to Lemma~\ref{f dyn pc alt}, $x_i(t+1) \geq f_i(x(t))$. For the arbitrary but fixed $i$, we construct the vector $y$ as
    \begin{equation}
        y_j =
        \begin{cases}
            z^*+\epsilon    & \text{if } j=i\\
            z^*-\epsilon & \text{if } j \in \B,~j \neq i\\
            -1  & \text{if } i \not\in \B.
        \end{cases}
    \end{equation}
    Since $x(t) \geq y$ and $f$ is non-decreasing according to Lemma~\ref{increasing_model}, $f_i(x(t)) \geq f_i(y)$, and consequently, $x_i(t+1) \geq f_i(y)$. Thus, it is sufficient to show that $f_i(y) \geq z^*+\alpha$. Hence, we write
    \begin{align}
    \label{znew2}
        f_i(y)
        & = y_i + \sum_{j \in \V} d'_i(y_i,y_j)w_{ij}(t)|y_j|(\sgn(y_j)-y_i)\cr
        & = z^*+\epsilon
        + \left(\sum_{j \in \B} w_{ij}(t)\right)(z^*-\epsilon)(1-(z^*+\epsilon))\cr
        &+\hat{d}_i(z^*+\epsilon) \left(\sum_{j \in \V\backslash\B} w_{ij}(t)\right)(-1-(z^*+\epsilon))\cr
        & \geq z^* +\epsilon + \left(\sum_{j \in \B} w_{ij}(t)\right)\cr &\times\left[ (z^*-\epsilon)(1-(z^*+\epsilon)) + \frac{\hat{d}_i(z^*+\epsilon)}{\gamma_\B} (-1-(z^*+\epsilon)) \right]\cr
        & \geq z^*+\epsilon,
    \end{align}
    where the last inequality in \eqref{znew2} is a result of
    \begin{equation}
        \left[ (z^*-\epsilon)(1-(z^*+\epsilon)) + \frac{\hat{d}_i(z^*+\epsilon)}{\gamma_\B} (-1-(z^*+\epsilon)) \right] > 0,
    \end{equation} which itself is implied from \eqref{small espsilon} considering $\hat{d}_i(z^*-\epsilon)\leq \hat{d}_i(z^*-\epsilon)$ according to Lemma~\ref{increasing_model}.
    
    {\bf Step 6:} Combining Steps~4 and 5, we conclude that for each $i \in \B$, there is a time $t_i$ such that $x_i(t) \geq z^*+\epsilon$ for any $t \geq t_i$. Hence, $z(t) \geq z^*+\epsilon$ for any $t \geq \max(t_1,\ldots,t_n)$, which contradicts the assumption $\lim_{t\rightarrow\infty} z(t)=z^*$ made in Step~3, completing the proof.\qed

\subsection{Derivations of Propositions~\ref{prop:polarization_result}, \ref{polarization_result}, \ref{polarization_result_Second_model}, and \ref{semi-general PC result}}\label{rest of proofs}

    In this subsection, starting from Proposition~\ref{prop:polarization_result}, we demonstrate that each proposition stated in this paper can be derived from the one coming next, while the last proposition, that is Proposition~\ref{semi-general PC result}, is a result of Theorem~\ref{general PC result} proved previously.

    Condition \eqref{prop 1 condition} in Proposition~\ref{prop:polarization_result} indicates that $\B$ has an infinite bubble number. Thus, assuming that Proposition~\ref{polarization_result} is true, in view of \eqref{alphagamma}, we obtain $\alpha_1 = 0$ and $\alpha_2 = 1$, immediately resulting in Proposition~\ref{prop:polarization_result}. Setting $d=1$ in Proposition~\ref{polarization_result_Second_model} simply converts it into Proposition~\ref{polarization_result}. Proposition~\ref{polarization_result_Second_model} is a special case of Proposition~\ref{semi-general PC result} where $\beta \rightarrow 0$. Thus, it only remains to derive Proposition~\ref{semi-general PC result} from Theorem~\ref{general PC result}.

    The assumption in Proposition~\ref{semi-general PC result} that \eqref{alphagamma3} has two positive solutions $\alpha_1$ and $\alpha_2$ in $(0,1)$ means that for any $\alpha \in (\alpha_1,\alpha_2)$ we have
    \begin{equation}
        \frac{1+\alpha}{\alpha(1-\alpha)} < \frac{\gamma_\B}{1-(1-d)\alpha^{\beta}}.
    \end{equation}
    Thus, setting $\hat{d}_i(\alpha) = 1-(1-d)\alpha^\beta$, which is a non-increasing function in $(0,1)$, in Theorem~\ref{general PC result} immediately leads to the statement of Proposition~\ref{semi-general PC result}.

\bibliographystyle{IEEEtran}
\bibliography{references}

\end{document}